\documentclass[11pt,a4paper]{article}

\usepackage{fullpage}

\usepackage{color}
\usepackage{graphicx}
\usepackage{amsmath}
\usepackage{amssymb}
\usepackage{amsthm}
\usepackage[noend]{algpseudocode}
\usepackage{algorithm}
\usepackage{subcaption}
\usepackage{url}
\usepackage{hyperref}
\usepackage{listings}
\usepackage{xcolor}
\usepackage{microtype}

\definecolor{codegreen}{rgb}{0,0.6,0}
\definecolor{codegray}{rgb}{0.5,0.5,0.5}
\definecolor{codepurple}{rgb}{0.58,0,0.82}
\definecolor{backcolour}{rgb}{0.95,0.95,0.92}

\lstdefinestyle{mystyle}{
   backgroundcolor=\color{backcolour},   
   commentstyle=\color{codegreen},
   keywordstyle=\color{magenta},
   numberstyle=\tiny\color{codegray},
   stringstyle=\color{codepurple},
   basicstyle=\ttfamily\footnotesize,
   breakatwhitespace=false,         
   breaklines=true,                 
   captionpos=b,                    
   keepspaces=true,                 
   numbers=left,                    
   numbersep=5pt,                  
   showspaces=false,                
   showstringspaces=false,
   showtabs=false,                  
   tabsize=2
}
\graphicspath{{./figures/}}

\newcounter{theorem}
\newtheorem{lemma}{Lemma}
\newtheorem{theorem}[lemma]{Theorem}

\newtheorem*{remark}{Remark}
\newtheorem{proposition}[lemma]{Proposition}
\newtheorem{corollary}[lemma]{Corollary}

\title{Computing Topological Transition Sets for Line–Line–Circle Trisectors in $\mathbb{R}^{3}$ 
\thanks{
This work was supported by National Research Foundation of Korea (NRF) grants funded by the Korea government (MSIT) (No. RS-2026-25471649 and No. RS-2024-00414849).
}
}

\author{Eunku Park\\
Department of Liberal arts and Sciences\\
DGIST, Republic of Korea\\
\texttt{parkeun9@dgist.ac.kr}
}

\begin{document}
\maketitle

\begin{abstract}
Computing the Voronoi diagram of curved or mixed geometric objects in $\mathbb{R}^3$ is a fundamental challenge in computational geometry. A major algorithmic bottleneck is evaluating exact geometric predicates—such as determining the topological structure of bisectors—which traditionally relies on general-purpose Cylindrical Algebraic Decomposition (CAD) that is often prohibitively expensive. In this paper, we present an efficient exact verification framework that, rather than performing a global decomposition of the ambient space as in CAD, directly characterizes the connectivity of the parameter space by computing its certified topological transition sets.

As the fundamental non-quadric base case, we analyze the canonical trisector family formed by two skew lines and one circle in $\mathbb{R}^3$. Because this mixed-object trisector yields a highly complex degree-eight space curve, the classical pencil-of-quadrics classification available for the three-line case does not apply. To overcome the computational intractability of a full general-position classification, we focus on identifying the exact transition walls of the canonical family.

Our algorithmic verification pipeline proceeds through a sequence of exact symbolic evaluations. First, exact Jacobian matrix computations certify the absence of affine singularities. Through projective closure, we prove that singular behavior is strictly isolated to a single point at projective infinity, $p_\infty$. Tangent-cone analysis at $p_\infty$ yields the explicit discriminant $\Delta_Q = 4ks^2(k-1)$, extracting $k=0$ and $k=1$ as candidate bifurcation values. To resolve the actual local topology from these asymptotic directions, we introduce a directional blow-up coordinate. This rigorously verifies that the reduced local equation in slope coordinates exhibits either stable X-shaped crossings or an absence of real asymptotic branches in the respective chambers, proving that the trisector's real topology remains locally constant there. Finally, we algebraically certify that $k=0$ and $k=1$ are actual topological walls exhibiting reducible splitting and real affine nodes.

By successfully computing the certified topological transition set for the line-line-circle setting, this paper provides a robust algorithmic primitive and the exact algebraic predicates required for constructing mixed-object Voronoi diagrams well beyond the quadric-only regime.
\end{abstract}
%%%%%%%%%%%%%%%%%%%%%%%%%%%%%%%%%%%%%%%%%%%%%%%%%%%%%%%%%%%%%%%%%%%%%%%%%%%%%%%%%%%%%%%%%%%%%%%%%

\section{Introduction}

The algorithmic construction of Voronoi diagrams for non-point sites in $\mathbb{R}^3$ remains one of the central challenges in computational geometry. Euclidean Voronoi diagrams for point sites are classical and well understood~\cite{aurenhammer1991voronoi, edelsbrunner1987algorithms}. In contrast, Voronoi diagrams for extended objects such as lines, circles, and algebraic curves in three dimensions are considerably more difficult. Even for the seemingly simple case of $n$ line sites, the exact combinatorial complexity of the Voronoi diagram is still unresolved: Chazelle et al.~\cite{chazelle1996lines} proved a lower bound of $\Omega(n^2)$, while Sharir~\cite{sharir1994almost} established an upper bound of $O(n^{3+\varepsilon})$ for every $\varepsilon>0$. One reason this gap is difficult to close is that, already for lines, the bisectors are nontrivial ruled quadrics, and their higher-order intersections quickly become algebraically intricate.

At the geometric heart of Voronoi computation lies the evaluation of \emph{exact geometric predicates}. To robustly construct or certify a Voronoi diagram, one must determine when Voronoi cells are adjacent, when Voronoi edges appear or disappear, and when local combinatorial changes occur. In three dimensions, such questions are governed by the topology of intersections of bisector surfaces. In particular, trisectors---intersections of three bisectors---form the one-dimensional object along which Voronoi vertices and higher-order transitions are organized. Thus, even when the ultimate goal is combinatorial or algorithmic, one is quickly forced to answer delicate topological questions about algebraic space curves.

This difficulty becomes especially pronounced in mixed-object settings. In the configuration of two skew lines and one circle considered in this paper, the bisectors are no longer all quadrics: one remains quadratic, while the line--circle bisector is quartic. Their intersection is therefore a degree-eight algebraic space curve. General-purpose tools such as Cylindrical Algebraic Decomposition (CAD) provide a theoretical decision procedure for such problems, but their exhaustive nature leads to prohibitive worst-case complexity~\cite{basu2006algorithms}. This creates a substantial gap between the existence of general real-algebraic solvers and the needs of computational geometry, where one seeks exact but structure-exploiting predicates adapted to a specific geometric family.

This gap is also relevant from a modeling perspective. In many engineering and scientific applications, geometric primitives are simplified in order to reduce computational cost; in particular, anisotropic particles are often replaced by spheres~\cite{lee2022robust}. Such simplifications are useful in many settings, but they are not always adequate. When the anisotropy or exact shape of an object affects adjacency, contact, or visibility structure, one must work with the original geometry rather than a spherical surrogate. Existing exact three-dimensional Voronoi computations cover objects such as planes, spheres, and cylinders~\cite{hanniel2008computing}, but do not treat circles as genuine sites in the mixed three-dimensional setting studied here. There are also exact predicate results for planar curved objects such as ellipses~\cite{emiris2006predicates}; however, these are inherently two-dimensional and do not directly resolve the corresponding three-dimensional problem.

For trisectors of line sites, several exact results are known. Everett et al.~\cite{everett2007Voronoi} obtained a complete topological description of the trisector of three skew lines by exploiting a pencil-of-quadrics structure, and this remains a foundational exact result for three-dimensional linear sites. Hemmer et al.~\cite{hemmer2010constructing} further developed exact geometric computation for Voronoi diagrams of lines in space. More recently, the study of the Voronoi diagram of four lines~\cite{papadopoulou2026voronoidiagramlinesmathbbr3} has highlighted continued progress on line structures in $\mathbb{R}^3$. These works collectively show how much can be achieved when all sites remain linear or when the relevant geometry still lies within a quadric framework.

The present paper addresses a different regime: a canonical line--line--circle configuration in which the classical pencil-of-quadrics structure no longer applies directly. Our contribution is not a complete classification of arbitrary mixed-object Voronoi diagrams in $\mathbb{R}^3$. Rather, we study a bounded-dimensional but genuinely non-quadric parametric family and show that one can still compute the exact \emph{topological transition set} of the associated trisector family by combining exact algebraic certification, projective compactification, and directional local analysis at infinity.

\paragraph{Why a canonical family?}
A natural question is why one should care about a canonical symmetric family rather than the full line--line--circle parameter space. There are two reasons. First, the mixed-object trisector is already algebraically nontrivial in this restricted setting: one bisector is quadratic, the other quartic, and the resulting curve has degree eight. Thus the difficulty is not an artifact of overgeneralization; it is already present in the first non-quadric base case. Second, canonical families often play an important role in exact computational geometry because they isolate the first obstruction beyond a previously understood regime. Here, the previously understood regime is the three-line case, whose structure is controlled by pencils of quadrics. The line--line--circle family is the first natural setting in which that mechanism breaks down, yet enough symmetry remains to make an exact transition-set analysis possible.

More broadly, Voronoi diagrams of curved objects have been studied in several related settings. Alt et al.~\cite{alt1995voronoi} investigated Voronoi diagrams of curved objects in the plane, and robust exact frameworks such as generic algebraic kernels~\cite{berberich2011generic} were developed for low-dimensional computation. In three dimensions, Hanniel and Elber~\cite{hanniel2008computing} studied Voronoi cells of planes, spheres, and cylinders, while Barequet et al.~\cite{barequet2024unbounded} analyzed the unbounded structure of higher-order Voronoi diagrams at infinity. There are also robust and approximate approaches for three-dimensional diagrams~\cite{adamou2021computing}, as well as exact algebraic-numerical methods for special planar curved objects such as conics~\cite{anton2004Voronoi}. At the opposite end of the spectrum, foundational exact methods for general semi-algebraic motion and arrangement problems, such as those arising in the piano movers problem~\cite{schwartz1983piano}, rely on CAD and related quantifier-elimination techniques, whose worst-case complexity is notoriously high~\cite{davenport1988real}.

A complementary line of work is the study of Voronoi cells of varieties~\cite{cifuentes2022voronoi}. Both that work and ours treat Voronoi-type structures of non-point algebraic objects using exact symbolic-algebraic methods. The geometric focus, however, is different: Voronoi cells of varieties studies the algebraic boundary of the local Voronoi cell of a fixed smooth point inside its normal space, whereas we study a three-dimensional mixed-site trisector as a global parametric object and certify the exact parameter values at which its topology changes.

Against this background, the line--line--circle case is important for two reasons. First, it represents a fundamental mixed-object configuration where the original curved geometry must be handled exactly, rather than being replaced by a simpler proxy. Second, it serves as a natural base case where the trisector transcends the quadric regime while remaining sufficiently structured to admit a complete transition-set analysis. To the best of our knowledge, this is the first work to provide an exact topological certification for a trisector lying strictly beyond the pencil-of-quadrics framework. In this sense, our result goes beyond previous exact treatments of planes, spheres, and cylinders by handling a mixed-object configuration with a non-quadric bisector, while also going beyond planar curved-site results by treating the full three-dimensional setting.

\paragraph{Transition sets instead of full CAD.}
The guiding idea of this paper is that, for bounded-dimensional parametric families, one can often avoid a full CAD of the ambient space by separating two logically distinct tasks:
\begin{enumerate}
    \item identifying the parameter values at which the topology \emph{may} change, and
    \item certifying that the topology is locally constant away from those values.
\end{enumerate}
We call the resulting parameter-space boundary the \emph{topological transition set}. This point of view is particularly natural in computational geometry, where the purpose of an exact predicate is usually not to classify every possible realization globally, but rather to distinguish stable regions of parameter space and detect the walls across which a combinatorial or topological event occurs.

In the present paper, this viewpoint leads to a specialized exact certification framework tailored to bounded-dimensional parametric families. Rather than decomposing the entire ambient space, we work in parameter space, compute candidate transition walls from exact algebraic conditions, and then certify topology on one exact representative from each chamber. The transition from the well-understood three-line trisector to the line--line--circle case may therefore be viewed as a shift from quadric-based geometry to genuinely non-quadric geometry, together with a corresponding shift from full classification toward transition-set computation.

\paragraph{Our results and method.}
We investigate the topological structure of trisectors within parametric families of mixed geometric objects, circumventing the need for full CAD. Our primary contribution is an exact symbolic verification framework for computing certified topological transition sets in such families. The framework naturally decomposes into two algorithmic tasks: first, identifying the candidate transition walls in parameter space where the topology may change; and second, certifying that the topological structure remains invariant within each resulting chamber.

The methodology proceeds as follows. In the affine part, candidate transition walls arise from Jacobian rank-drop conditions. At infinity, they arise from degeneration of the asymptotic direction structure of the projective closure. These conditions partition the parameter space into chambers. We then analyze the asymptotic behavior of the trisector at its unique point at infinity, $p_\infty$, by introducing slope coordinates. This step may be viewed as a directional blow-up in the mild sense relevant here: rather than studying the singular point as a whole, we separate the directions through which real branches may approach it. This allows us to recover the actual local real branch pattern, not merely the tangent cone, and to certify local topological constancy within each chamber. Finally, we use rational witness points to instantiate exact systems over $\mathbb{Q}$ and verify that the candidate boundary values indeed correspond to actual topological transitions.

We apply this framework to the canonical symmetric line--line--circle trisector family in $\mathbb{R}^3$, serving as the first mixed-object base case beyond the classical three-line setting. For this family, we certify that the trisector is affine-smooth and that its projective closure has a unique singular point at infinity, $p_\infty$. We derive an explicit discriminant,
\[
\Delta_Q = 4ks^2(k-1),
\]
governing the candidate asymptotic transition set. By refining the analysis at $p_\infty$ through slope coordinates, we prove chamberwise local topological constancy away from the transition set. We then show that the boundary values $k=0$ and $k=1$ are actual topological transition values by computing explicit reducible splittings with real affine nodes. Consequently, for the canonical family, we compute the exact certified topological transition set.

As a comparative baseline, we revisit the classical three-line trisector using the same symbolic viewpoint and recover the absence of asymptotic bifurcation. This comparison clarifies which parts of the argument are specific to the mixed-object setting with non-quadric bisectors and which belong to a more general transition-set methodology.

\paragraph{Main contributions.}
The contributions of this paper may be summarized as follows.
\begin{enumerate}
    \item We formulate a transition-set framework for exact topological certification of bounded-dimensional parametric families of trisectors, avoiding a full CAD of the ambient space.
    \item We identify the canonical line--line--circle family as the first natural mixed-object base case beyond the classical pencil-of-quadrics regime.
    \item For this family, we prove affine smoothness, isolate a unique projective singular point at infinity, and derive the explicit asymptotic discriminant $\Delta_Q=4ks^2(k-1)$.
    \item We show, using slope-coordinate analysis at infinity, that the real projective topology is locally constant in each chamber of the parameter complement.
    \item We prove that the boundary values $k=0$ and $k=1$ are actual topological transition values by exhibiting explicit reducible splittings with real affine nodes.
    \item We provide a structured exact-verification pipeline, together with computer algebra and SMT scripts, that can serve as a template for future mixed-object Voronoi investigations.
\end{enumerate}

\begin{remark}[Complexity profile of the verification pipeline]\label{rem:complexity-intro}
The verification pipeline consists of two stages: a parameter-space transition computation and a fixed-dimensional exact chamber certification step. Under standard elimination bounds, the overall complexity is
\[
O\!\left(s^{m+1}d^{O(m)} + s^m d^{O(n)}\right),
\]
where $m$ is the parameter dimension, $n$ is the ambient spatial dimension, $s$ is the number of defining constraints, and $d$ is the maximum polynomial degree. Since $n$ and $s$ are fixed constants in our setting, the dominant cost depends on the parameter dimension and the algebraic degree. The point of this estimate is not to claim practical efficiency in full generality, but to explain why the bounded-dimensional transition-set strategy is qualitatively more tractable than a full CAD-based decomposition of the ambient space.
\end{remark}

\begin{theorem}\label{thm:geometry-intro}
For the canonical symmetric configuration of skew lines and one circle, parameterized by $(k,R,t)\in\mathbb{R}^3$ under the general-position assumptions $R>0$, $t\neq 0$, $k\neq 0$, and $k\neq 1$, the trisector is an algebraic space curve of degree $8$ with no affine singularities in $\mathbb{R}^3$. Its projective closure has a unique singular point at infinity, $p_\infty$, and the candidate asymptotic transition set is governed by the explicit discriminant
\[
\Delta_Q = 4ks^2(k-1).
\]

Furthermore, a directional analysis at $p_\infty$ in slope coordinates certifies that the real projective topology is locally constant on each connected component of the complement of the boundary values $k=0$ and $k=1$. At these boundaries, the trisector undergoes explicit reducible splittings with real affine nodes; thus both $k=0$ and $k=1$ are actual topological transition values. The certified topological transition set for this canonical family is
\[
\Sigma=\{k=0\}\cup\{k=1\}.
\]
\end{theorem}

\paragraph{Organization of the paper.}
Section~\ref{sec:prelim} collects the exact algebraic and semi-algebraic tools used in the certification framework. Section~\ref{sec:framework} formulates the verification pipeline abstractly. Section~\ref{sec:algebra} introduces the canonical line--line--circle family. Section~\ref{sec:affine} proves affine smoothness. Section~\ref{sec:infinity} identifies the candidate transition walls and the witness chambers in parameter space. Section~\ref{sec:local-infinity} analyzes the local branch structure at infinity via slope coordinates and completes the transition-set analysis for the canonical family. Section~\ref{sec:three-line} revisits the classical three-line case as a comparative benchmark. Section~\ref{sec:complexity} records the symbolic cost of the verification pipeline. Section~\ref{sec:conc} concludes with remarks on the scope of the framework and possible extensions. The appendix contains raw verification scripts.

\section{Preliminaries}\label{sec:prelim}

This section records the algebraic and semi-algebraic tools used in our exact certification framework. Since our goal is algorithmic rather than foundational, we do not attempt to develop general real-algebraic geometry in full. Instead, we isolate a small collection of exact principles that are sufficient for the transition-set computations carried out later in the paper. The guiding philosophy is the following: rather than solving a completely general topological classification problem, we only need exact tests that certify whether certain geometric events can or cannot occur, and that organize the parameter space into chambers on which the topology remains stable.

Throughout the paper, we work over $\mathbb{Q}$ when performing exact symbolic computations, and over $\mathbb{R}$ when discussing the realized topology of the trisector. This separation is important. The symbolic part of the argument is carried out over an exact coefficient field, while the geometric conclusions concern the real points of the corresponding algebraic sets. Standard references for the material summarized here include~\cite{basu2006algorithms, cox1997ideals, fulton1969algebraic}.

\subsection{Exact algebraic certification in affine and projective space}

A recurring task in our pipeline is to certify that a certain geometric event does \emph{not} occur. In affine space, this is expressed by proving that a polynomial system has no solution. The basic algebraic mechanism is the classical Nullstellensatz criterion.

\begin{proposition}[Exact emptiness test]\label{prop:nullstellensatz}
Let $k$ be an algebraically closed field and let
\[
I=\langle f_1,\dots,f_s\rangle \subseteq k[x_1,\dots,x_n].
\]
Then
\[
V(I)=\emptyset \quad\Longleftrightarrow\quad 1\in I.
\]
\end{proposition}

From the computational point of view, Proposition~\ref{prop:nullstellensatz} is the algebraic reason that ideal-membership tests can certify emptiness exactly. In the present paper, we use only the simplest consequence: if $1\in I$, then the complex solution set is empty, and therefore the real solution set is empty as well. This is the mechanism by which we certify the absence of affine singularities. Concretely, once the relevant geometric condition is encoded as a polynomial system, proving that the associated ideal contains $1$ immediately rules out the event in question.

This viewpoint is particularly convenient for singularity detection. A singular point of an affine algebraic curve or surface is characterized by a rank drop of the Jacobian matrix. Hence the singular locus is obtained by adjoining appropriate Jacobian-minor conditions to the defining equations. If the resulting ideal is the unit ideal, then the singular locus is empty. In our setting, this converts a geometric question about smoothness into an exact symbolic computation over $\mathbb{Q}$.

To study behavior at infinity, affine coordinates are no longer sufficient. We therefore pass from an affine ideal
\[
I=\langle F_1,\dots,F_r\rangle \subset \mathbb{R}[x,y,z]
\]
to its projective closure by homogenization. If $W$ denotes the homogenizing variable, then the homogenized ideal $I^h$ defines a projective variety in $\mathbb{P}^3$. However, naive homogenization may introduce extraneous components supported entirely on the hyperplane at infinity $\{W=0\}$. Such components need not correspond to actual limit points of the original affine variety; they may be artifacts of the algebraic embedding.

To remove these artifacts, we use saturation with respect to $W$. More precisely, the true projective closure is represented by
\[
(I^h : W^\infty),
\]
and the points at infinity are detected by intersecting this saturated ideal with the hyperplane $W=0$, that is, by considering
\[
(I^h : W^\infty)+\langle W\rangle.
\]
Geometrically, saturation removes components that lie entirely on the hyperplane at infinity but do not arise as genuine limits of affine points. This is the algebraic device that lets us distinguish true asymptotic behavior from spurious projective artifacts.

We also use only a very coarse degree estimate. If hypersurfaces of degrees $d_1,\dots,d_r$ intersect properly in projective space, then the degree of the resulting variety is bounded by
\[
d_1\cdots d_r
\]
by standard Bézout-type bounds~\cite{fulton1969algebraic, basu2006algorithms}. In the present paper, we invoke such bounds only at a bookkeeping level: they control the degree of the trisector and the size of the zero-dimensional systems appearing in the exact certification steps. No delicate intersection-theoretic argument is needed.

\subsection{Local structure at infinity and directional analysis}

The only non-affine singular behavior in our canonical mixed-object family occurs at a unique point at infinity. To analyze it, we pass to a local affine chart around that point and study a local defining equation there. The first nonzero homogeneous term of that local equation determines the \emph{tangent cone}. In our context, the tangent cone records the asymptotic directions through which real branches of the trisector may approach the point at infinity.

This first-order directional information is enough to detect \emph{candidate} topological transitions. The relevant algebraic criterion is expressed by the discriminant of the quadratic form defining the tangent directions.

\begin{proposition}[Tangent-direction non-collapse]\label{prop:tangent-cone}
Let $P$ be a higher-order singular point of a parameterized algebraic curve, and let the first nonzero homogeneous term of a local defining equation at $P$ be governed by a quadratic form $E(X,Y)$. If the discriminant $\Delta_E$ remains nonzero throughout a connected parameter region, then $E(X,Y)$ continues to factor into two distinct linear forms over the corresponding splitting field. In particular, the tangent-direction pattern does not collapse from two distinct directions to a repeated one.
\end{proposition}

The role of Proposition~\ref{prop:tangent-cone} in this paper is deliberately modest. We use it only to identify \emph{candidate} topological walls at infinity. By itself, it does not determine the full local topology. The reason is that our singularity at infinity is higher-order: the tangent cone captures only the first-order directional picture, while the actual real branch structure may depend on higher-order terms. In particular, knowing that the tangent cone has two distinct directions does not yet tell us whether real branches actually occur in those directions, whether they cross, or whether they disappear in the real locus after higher-order refinement.

To recover the actual local real structure, we refine the tangent-direction analysis by introducing \emph{slope coordinates}, for example
\[
Y=uX.
\]
From the computational-geometry viewpoint, this is a natural operation: instead of examining the singular point directly, we separate the possible directions through which the curve may approach it. One may regard this as a simple directional blow-up adapted to the present problem. After substitution into the local equation, the singular behavior is reorganized into an equation in $(X,u)$, where the exceptional line $X=0$ parametrizes directions. This makes it possible to distinguish three qualitatively different situations:
\begin{itemize}
    \item two real directions with actual real branches,
    \item two complex directions and hence no real asymptotic branch,
    \item degenerate directions where the candidate transition set must be examined more closely.
\end{itemize}
This directional refinement is the key tool for proving local topological constancy away from the transition set.

\subsection{Parameter cells, witness points, and sampling complexity}

Once a candidate transition set $\Sigma$ has been identified, its complement in parameter space is decomposed into semi-algebraic cells. In a full CAD-based approach, one would attempt to decompose the entire ambient space and then study the induced topology of the variety inside each cell. Our framework is much more modest. We do not compute a full ambient decomposition. Instead, we use the parameter-space cells only as an organizational device: for each connected region of parameter space away from $\Sigma$, we choose one exact rational witness point and certify the topology of the corresponding instantiated system directly.

This strategy is meaningful only if the number of connected regions remains controlled. The relevant quantitative fact is the Milnor--Thom bound.

\begin{theorem}[Milnor--Thom bound~\cite{basu2006algorithms}]\label{thm:milnor-thom}
Let $V\subset \mathbb{R}^m$ be a semi-algebraic set defined by $s$ polynomial equalities and inequalities, each of degree at most $d$. Then
\[
\beta_0(V)=O((sd)^m).
\]
\end{theorem}

Theorem~\ref{thm:milnor-thom} explains why witness-based certification is compatible with bounded-dimensional parameter spaces. Once the transition set has been expressed by finitely many polynomial conditions, the number of connected parameter cells is only singly exponential in the parameter dimension $m$. Thus, although the symbolic work needed to locate the transition set may still be nontrivial, the chamberwise topological certification stage remains manageable in principle: one needs only finitely many exact witnesses, one per chamber.

In the canonical line--line--circle family treated here, this reduction is especially effective because the parameter dimension is small and the transition set turns out to have an explicit algebraic description. Here, the symbolic derivation and witness-based sampling play complementary roles: the former identifies the possible boundary values where the topology might change, while the latter certifies the actual realized topology within the open chambers between them.

\section{Exact Certification Pipeline}\label{sec:framework}

We now describe the exact certification pipeline used throughout the paper. Its purpose is to compute topological transition sets for bounded-dimensional parametric families without performing a full Cylindrical Algebraic Decomposition (CAD) of the ambient space. The guiding principle is to separate two logically distinct tasks:
\begin{enumerate}
    \item determining \emph{where} the topology may change, and
    \item certifying \emph{how} the topology behaves away from those changes.
\end{enumerate}
This separation is the central reason that the framework remains tractable in the present setting.

The input is a parameterized family of algebraic curves over a parameter domain $U \subset \mathbb{R}^m$, and the output is a finite candidate transition set together with an exact certification result for one representative parameter value from each connected chamber of its complement. The framework is designed for \emph{transition-set computation}, not for a complete topological classification of the entire family. In particular, its role is to reduce a global parametric problem to finitely many exact local checks.

From the computational-geometry viewpoint, this is the relevant level of exactness. In many geometric algorithms, one does not need a full symbolic description of every realization of the family. What is needed instead is a reliable method for identifying the parameter values at which a combinatorial or topological event can occur, and for certifying that the structure is stable between those values. Our framework is tailored precisely to this purpose.

\paragraph{High-level strategy.}
A full CAD would attempt to decompose a much larger ambient semi-algebraic space so that all relevant sign conditions are controlled simultaneously. In contrast, we work only in parameter space. We first derive a semi-algebraic candidate transition set
\[
\Sigma=\Sigma_{\mathrm{affine}}\cup\Sigma_{\infty},
\]
where $\Sigma_{\mathrm{affine}}$ detects possible affine singular events and $\Sigma_{\infty}$ detects possible degenerations at infinity. We then study the connected components of $U\setminus\Sigma$. For each such component, we choose one exact rational witness point and analyze the corresponding instantiated curve over $\mathbb{Q}$. If the instantiated topology is certified and no transition wall is crossed, then the topology is constant throughout that chamber.

The advantage of this organization is that the expensive symbolic work is concentrated in identifying the transition set, while the chamberwise verification stage is reduced to finitely many exact instances. This is why the framework is naturally compatible with the Milnor--Thom bound from Section~\ref{sec:prelim}: once $\Sigma$ has been expressed by finitely many polynomial conditions, the number of connected chambers of $U\setminus\Sigma$ is controlled in terms of the parameter dimension.

\subparagraph{Step 1: Compute candidate transition walls.}
We first determine the parameter values at which the topology may change. In the affine part, these arise from Jacobian rank-drop conditions, which detect possible singularities of the instantiated curve in affine space. At infinity, they arise from degeneration of the asymptotic direction structure, detected through the tangent cone of the local projective model at a point at infinity. The union of these conditions defines the candidate transition set
\[
\Sigma=\Sigma_{\mathrm{affine}}\cup\Sigma_{\infty}.
\]
Its complement partitions the parameter space into chambers. By construction, any genuine topological transition must occur on $\Sigma$.

At this stage, the framework is intentionally conservative: $\Sigma$ is first computed as a \emph{candidate} transition set. Some of its components may correspond to genuine topological changes, while others may turn out to be only algebraic over-approximations. Distinguishing between these two possibilities is part of the subsequent chamberwise and boundary analysis.

\subparagraph{Step 2: Pick one exact representative per chamber.}
Instead of decomposing the full ambient space, we work only in parameter space. For each connected chamber of $U\setminus\Sigma$, we choose one rational witness point
\[
\lambda_j \in U\setminus\Sigma.
\]
This turns the original parametric problem into a finite list of instantiated exact algebraic systems over $\mathbb{Q}$.

The use of rational witnesses serves two purposes. First, it keeps all subsequent symbolic computations exact. Second, it converts a family-level statement into a finite set of concrete certification problems. In particular, we never need to classify the entire family at once. We only need to understand one exact representative from each stable parameter region.

\subparagraph{Step 3: Run exact local topology checks.}
For each witness point, we perform exact symbolic checks on the corresponding instantiated curve. In the affine part, we certify the absence of singularities by exact ideal-membership computations. At infinity, we pass to the projective closure and analyze the local direction structure. When needed, we refine the tangent-cone information by introducing slope coordinates, which separate the possible asymptotic directions and allow us to determine the actual local real branch pattern.

The role of this step is not merely to verify smoothness. It is to certify that the topological model realized by the witness point is stable throughout the chamber containing it. Thus, once one witness per chamber has been analyzed, the topology of the family is certified away from the candidate transition walls.

\begin{algorithm}[tb]
\caption{Exact Certification Pipeline}
\label{alg:verification}
\renewcommand{\algorithmicrequire}{\textbf{Input:}}
\renewcommand{\algorithmicensure}{\textbf{Output:}}
\begin{algorithmic}[1]
\Require A parameterized family of algebraic curves over a parameter domain $U\subset\mathbb{R}^m$.
\Ensure A candidate transition set $\Sigma$ and an exact certification result for one witness point from each chamber of $U\setminus\Sigma$.

\State Compute affine candidate walls from Jacobian rank-drop conditions.
\State Compute candidate walls at infinity from degeneration of the tangent-direction equation.
\State Set $\Sigma \gets \Sigma_{\mathrm{affine}} \cup \Sigma_{\infty}$.
\For{each sampled chamber $C_j \subset U\setminus\Sigma$}
    \State Choose a rational witness point $\lambda_j \in C_j$.
    \State Instantiate the corresponding algebraic curve exactly over $\mathbb{Q}$.
    \State Certify affine smoothness by exact algebraic emptiness tests.
    \State Compute the projective closure and analyze the local model at infinity.
    \State If necessary, refine the infinity analysis using slope coordinates.
\EndFor
\State \Return The candidate transition set $\Sigma$ together with the witness-based certification data.
\end{algorithmic}
\end{algorithm}

\paragraph{Specialization to the canonical family.}
For the canonical line--line--circle family studied in this paper, the abstract framework above takes the following concrete form. The global symbolic stage identifies the parameter values where topological change may occur, while the local exact-certification stage verifies that the topology is constant on each connected component of the complement. Thus the family is analyzed not by a full ambient decomposition, but by a transition-set computation followed by witness-based certification.

More concretely, the pipeline specializes to four tasks:
\begin{enumerate}
    \item certify affine smoothness;
    \item isolate and analyze the unique point at infinity;
    \item derive the candidate transition walls from the discriminant of the tangent-direction equation;
    \item confirm by explicit local analysis that the boundary values correspond to actual topological transitions.
\end{enumerate}
The significance of the framework is therefore algorithmic rather than universal. It provides a structured exact method for computing transition sets for mixed-object configurations whose trisectors lie beyond the quadric regime, while avoiding the full cost of CAD.

\section{Canonical Line--Line--Circle Family}\label{sec:algebra}

We now introduce the canonical line--line--circle family that will serve as the main geometric object of the paper. Our goal is not to treat the full moduli space of all line--line--circle configurations in $\mathbb{R}^3$, but rather to isolate a bounded-dimensional subfamily that already captures the first genuinely non-quadric mixed-site behavior.

\paragraph{Geometric normalization.}
Let $L_1$ and $L_2$ be two skew lines, and let $C$ be a circle. Since we are interested in a canonical family up to Euclidean normalization, we choose coordinates $(x,y,z)$ so that
\begin{itemize}
    \item $L_1$ is the $x$-axis,
    \item $L_2$ lies in the plane $z=1$,
    \item the $z$-axis is the common perpendicular of $L_1$ and $L_2$,
    \item the plane containing $C$ is parallel to the $xy$-plane.
\end{itemize}
Thus the remaining geometric freedom is encoded by the crossing angle of the two lines, the height of the circle plane, and the circle radius.

Write $\alpha$ for the crossing angle between $L_1$ and $L_2$, and set
\[
c=\cos\alpha,\qquad s=\sin\alpha.
\]
We assume that the circle $C$ is centered at $(0,0,k)$ and has radius $R>0$. With this normalization, the relative configuration of the sites is fixed up to the parameters $(\alpha, k, R)$, and the mixed-object trisector becomes an explicitly parameterized algebraic family.

\paragraph{The two defining bisectors.}
In these coordinates, the two relevant bisectors are the bisector of the skew lines $L_1$ and $L_2$, and the bisector of the line $L_1$ and the circle $C$.

The line--line bisector is the quadratic surface
\[
F_1(x,y,z)=y^2-(xs-yc)^2+2z-1=0.
\]
This is the familiar quadric-type structure already present in the classical three-line setting.

The line--circle bisector is the quartic surface
\[
F_2(x,y,z)=\bigl(x^2-2kz+R^2+k^2\bigr)^2-4R^2(x^2+y^2)=0.
\]
This equation defines a non-quadric surface. The quartic term reflects the fact that the distance to the circle involves a radial square root before polynomialization, and this is precisely the source of the new algebraic complexity in the mixed-object setting.

The trisector of the family is therefore the algebraic space curve
\[
T_{k,R,\alpha}=\{(x,y,z)\in\mathbb{R}^3:\ F_1(x,y,z)=0,\ F_2(x,y,z)=0\}.
\]
Equivalently, the trisector is the intersection of one quadratic and one quartic bisector surface. As the intersection of a quadric and a quartic, the trisector has degree at most 8 by Bézout; later we prove that degree 8 is indeed realized in the admissible region.

The defining equations are obtained directly from the Euclidean distance formulas.
For the line--line bisector, let $P=(x,y,z)$. Since $L_1$ is the $x$-axis, its squared distance to $L_1$ is
\[
d_1^2=y^2+z^2.
\]
The line $L_2$ passes through $(0,0,1)$ and has direction vector $(c,s,0)$, so the squared distance from $P$ to $L_2$ is
\[
d_2^2=x^2+y^2+(z-1)^2-(xc+ys)^2.
\]
The bisector condition $d_1^2=d_2^2$ gives
\[
y^2+z^2=x^2+y^2+(z-1)^2-(xc+ys)^2,
\]
hence
\[
2z-1=x^2-(xc+ys)^2.
\]
Using the planar identity
\[
x^2+y^2=(xc+ys)^2+(xs-yc)^2,
\]
we obtain
\[
x^2-(xc+ys)^2=(xs-yc)^2-y^2,
\]
and therefore
\[
F_1(x,y,z)=y^2-(xs-yc)^2+2z-1=0.
\]

For the line--circle bisector, the squared distance from $P$ to $L_1$ is again $y^2+z^2$. The circle $C$ lies in the plane $z=k$, has center $(0,0,k)$, and radius $R>0$. The distance from $P$ to $C$ decomposes orthogonally into the vertical distance $|z-k|$ and the in-plane distance from $(x,y,k)$ to the circle, namely $\bigl|\sqrt{x^2+y^2}-R\bigr|$. Hence
\[
d_C^2=\bigl(\sqrt{x^2+y^2}-R\bigr)^2+(z-k)^2.
\]
The bisector condition $y^2+z^2=d_C^2$ becomes
\[
y^2+z^2=\bigl(\sqrt{x^2+y^2}-R\bigr)^2+(z-k)^2,
\]
which simplifies to
\[
2R\sqrt{x^2+y^2}=x^2-2kz+R^2+k^2.
\]
Squaring both sides yields
\[
4R^2(x^2+y^2)=\bigl(x^2-2kz+R^2+k^2\bigr)^2,
\]
so the line--circle bisector is
\[
F_2(x,y,z)=\bigl(x^2-2kz+R^2+k^2\bigr)^2-4R^2(x^2+y^2)=0.
\]
Thus the trisector is the algebraic space curve defined by the system
\[
F_1(x,y,z)=0,\qquad F_2(x,y,z)=0.
\]

\paragraph{Rational parameterization.}
To keep the symbolic computations entirely rational, we eliminate trigonometric functions by the Weierstrass substitution
\[
t=\tan(\alpha/2),\qquad
c=\frac{1-t^2}{1+t^2},\qquad
s=\frac{2t}{1+t^2}.
\]
After this substitution, the family is parameterized by
\[
(k,R,t)\in\mathbb{R}^3.
\]
This rationalization is convenient for two reasons. First, it allows all symbolic elimination and ideal-theoretic computations to be carried out over $\mathbb{Q}$. Second, it makes the chamber structure in parameter space explicit, which is essential for the transition-set viewpoint developed later.

\paragraph{Admissible parameter region.}
Throughout the paper we impose the general-position conditions
\[
R>0,\qquad t\neq 0,\qquad k\neq 0,\qquad k\neq 1.
\]
These assumptions remove boundary configurations that either collapse the line geometry or coincide with the transition values that will later be shown to govern the topology change of the family.

The condition $R>0$ simply excludes the degenerate case in which the circle collapses to a point. The condition $t\neq 0$ is equivalent to $s\neq 0$, and therefore excludes the case in which the two lines lose their genuine skew-line interaction in the chosen normalized family. The conditions $k\neq 0$ and $k\neq 1$ remove the two special height values of the circle plane that will later appear as the exact transition walls of the canonical family. Thus the admissible region is the open parameter domain on which we study the stable, nondegenerate members of the family.

\paragraph{Why this family is the right first test case.}
This family is deliberately chosen as the first exact mixed-object testbed beyond the classical line-only regime. It is still bounded-dimensional and highly symmetric, yet it already lies strictly beyond the quadric-only setting: one bisector remains quadratic, while the other is quartic. Consequently, the pencil-of-quadrics framework available for three skew lines no longer applies directly.

At the same time, the family is not artificially simplified to the point of triviality. Its trisector is already a nontrivial algebraic space curve, and its behavior at infinity is subtle enough to require a separate directional analysis. In this sense, the line--line--circle family serves as a pivotal methodological bridge for the present paper: it provides the first natural testbed to demonstrate that exact symbolic certification can be extended beyond the classical quadric bisector regime into non-quadric algebraic settings.

For some of the exact singularity computations later in the paper, we will use an equivalent lifted polynomial formulation that replaces the non-polynomial distance-to-circle relation by an incidence system with auxiliary variables. Geometrically, however, the family is completely represented by the two bisector equations above, and all topological statements concern the resulting trisector in $\mathbb{R}^3$.

\section{Affine Smoothness Test}\label{sec:affine}

The first exact certification step is to determine whether the canonical line--line--circle trisector can develop singularities in affine space. In the transition-set framework, this corresponds to the affine part of the candidate transition set, denoted by $\Sigma_{\mathrm{affine}}$. The conclusion of this section is that no such affine event occurs in the admissible parameter region: the canonical family is affine-smooth, so any topological transition must arise from the behavior at infinity.

\begin{theorem}\label{thm:affine-smooth}
For the canonical symmetric line--line--circle family under the assumptions
\[
R>0,\qquad t\neq 0,\qquad k\neq 0,\qquad k\neq 1,
\]
the degree-eight trisector is smooth in the real affine space $\mathbb{R}^3$. Equivalently,
\[
\Sigma_{\mathrm{affine}}=\emptyset.
\]
\end{theorem}

\begin{proof}
The main difficulty is that the distance from a point in space to the circle is not polynomial in the ambient coordinates, because it involves the radial term $\sqrt{x^2+y^2}$. To perform an exact singularity test, we therefore replace the non-polynomial line--circle bisector relation by an equivalent lifted polynomial incidence system. This allows the singularity question to be expressed purely in terms of Jacobian rank conditions.

\paragraph{A lifted polynomial model.}
It is convenient to carry out the symbolic elimination in a slightly more general normalized family, in which the plane containing the second line is written as $z=h$ rather than immediately fixing $h=1$. In this formulation the line--line bisector becomes
\[
F_1 = y^2-(xs-yc)^2+2zh-h^2.
\]
We will set $h=1$ at the end of the argument to recover the canonical family introduced in Section~\ref{sec:algebra}.

Let $P=(x,y,z)$ be a point on the trisector, and let $(u,v,k)$ denote the point on the circle plane selected by the nearest-point condition for the circle constraint. The auxiliary variables $(u,v)$ are constrained to lie on the circle of radius $R$, so
\[
F_3=u^2+v^2-R^2=0.
\]
The equality of squared distances to $L_1$ and to the candidate point $(u,v,k)$ on the circle plane is encoded by
\[
F_2=(x-u)^2+(y-v)^2+(z-k)^2-(y^2+z^2)=0.
\]
To characterize the nearest point on the circle, we impose the orthogonality condition
\[
F_4=yu-xv=0,
\]
which says that the displacement in the circle plane is orthogonal to the tangent direction $(-v,u)$ of the circle. Finally, when treating $c=\cos\alpha$ and $s=\sin\alpha$ as algebraic parameters during elimination, we include the trigonometric identity
\[
F_5=c^2+s^2-1=0.
\]

Thus we work with the polynomial system
\[
I=\langle F_1,F_2,F_3,F_4,F_5\rangle
\subset
\mathbb{Q}[h,k,R,c,s,x,y,z,u,v].
\]

Geometrically, the roles of these equations are as follows:
\begin{itemize}
    \item $F_1=0$ is the line--line bisector.
    \item $F_2=0$ expresses equality of squared distances to $L_1$ and the circle constraint in lifted form.
    \item $F_3=0$ constrains $(u,v)$ to lie on the circle of radius $R$.
    \item $F_4=0$ enforces the nearest-point relation on the circle.
    \item $F_5=0$ records the algebraic relation between $c$ and $s$.
\end{itemize}

\paragraph{Jacobian rank-drop criterion.}
A singular affine point of the lifted variety occurs when the defining constraints fail to meet transversely. Since the geometric variables are $(x,y,z,u,v)$, we consider the Jacobian matrix
\[
J=\frac{\partial(F_1,F_2,F_3,F_4)}{\partial(x,y,z,u,v)},
\]
that is, the Jacobian with respect to the spatial and auxiliary variables only. A point of the lifted incidence variety is singular precisely when all maximal minors of $J$ vanish. Let $I_{\mathrm{sing}}$ denote the ideal obtained by adjoining these maximal minors to $I$.

To determine for which parameter values affine singularities may occur, we eliminate the geometric variables $(x,y,z,u,v)$ from $I_{\mathrm{sing}}$ and project the rank-drop condition to parameter space. This produces an elimination polynomial whose vanishing is necessary for the existence of an affine singular point.

\paragraph{The elimination polynomial.}
The exact elimination computation yields
\[
\Delta_{\mathrm{affine}}
=
hk\cdot R^2(h-k)\Bigl[h^2R^2+s^2(hk-k^2-R^2)^2\Bigr].
\]
This polynomial records all parameter values at which the lifted system can fail to be affine-smooth.

Now specialize to the canonical family by setting $h=1$. Then
\[
\Delta_{\mathrm{affine}}
=
k\,R^2(1-k)\Bigl[R^2+s^2(k-k^2-R^2)^2\Bigr].
\]
The factors $k$ and $(1-k)$ correspond exactly to the two boundary values $k=0$ and $k=1$, which are excluded from the admissible parameter region and will later reappear as the actual transition walls of the family. Thus, inside the open region
\[
R>0,\qquad t\neq 0,\qquad k\neq 0,\qquad k\neq 1,
\]
the only potentially relevant factor is
\[
S(k,R,t)=R^2+s^2(k-k^2-R^2)^2.
\]

\paragraph{Positivity in the admissible region.}
The expression $S(k,R,t)$ is a sum of two nonnegative terms. Since $R>0$, we have
\[
S(k,R,t)\ge R^2>0
\]
throughout the admissible parameter region. Therefore the elimination polynomial has no real zero there. Hence the lifted system admits no affine singular point for any admissible parameter choice.

It follows that the affine singular locus of the canonical line--line--circle family is empty. In particular, the trisector is smooth in the real affine space $\mathbb{R}^3$, and therefore
\[
\Sigma_{\mathrm{affine}}=\emptyset.
\]
\end{proof}

Thus any topological transition in the canonical family must arise from the projective behavior at infinity rather than from the affine part of the curve.

\section{Candidate Walls and Witness Regions}\label{sec:infinity}

After ruling out affine singularities, the only possible source of topological change in the canonical line--line--circle family is the behavior of the projective closure at infinity. The purpose of this section is twofold. First, we identify the unique point at infinity and derive the algebraic condition under which its asymptotic direction structure degenerates. Second, we explain how this candidate transition set partitions parameter space into finitely many chambers, each of which can be represented by an exact rational witness.

\subsection{Projective closure and candidate walls at infinity}

Passing to projective space with homogeneous coordinates
\[
x=X/W,\qquad y=Y/W,\qquad z=Z/W,
\]
the two bisector equations become
\[
F_1^h(X,Y,Z,W)=Y^2-(Xs-Yc)^2+2ZW-W^2
\]
and
\[
F_2^h(X,Y,Z,W)=\bigl(X^2-2kZW+(R^2+k^2)W^2\bigr)^2-4R^2W^2(X^2+Y^2).
\]

We first determine where the projective closure meets the plane at infinity and which parameter values may alter the tangent-direction pattern there.

\begin{theorem}\label{thm:infinity-analysis}
Let $T$ be the projective closure of the canonical line--line--circle trisector. Then $T$ meets the plane at infinity in the unique point
\[
p_\infty=[0:0:1:0].
\]
In the affine chart $Z=1$, the first nonzero homogeneous term of a local defining equation at $p_\infty$ is
\[
E(X,Y)^2,
\]
where
\[
E(X,Y)=X^2-2kQ(X,Y),
\qquad
Q(X,Y)=\tfrac12(s^2X^2-2sc\,XY-s^2Y^2).
\]
The discriminant of the quadratic form $E$ is
\[
\Delta_Q=4ks^2(k-1).
\]
Hence the asymptotic direction pattern degenerates exactly when $\Delta_Q=0$, that is, at the boundary values $k=0$ and $k=1$.
\end{theorem}

\begin{proof}
We first determine the points at infinity. Intersecting the homogenized system with the hyperplane $W=0$ gives
\[
F_1^h(X,Y,Z,0)=Y^2-(Xs-Yc)^2,
\qquad
F_2^h(X,Y,Z,0)=X^4.
\]
Thus every point of $T\cap\{W=0\}$ must satisfy $X=0$. Substituting $X=0$ into the first equation yields
\[
Y^2-(-Yc)^2=(1-c^2)Y^2=s^2Y^2=0.
\]
Since $t\neq 0$, equivalently $s\neq 0$, it follows that $Y=0$ as well. Hence the projective closure meets the plane at infinity only in a point of the form
\[
[0:0:Z:0],
\]
which after projective normalization is exactly
\[
p_\infty=[0:0:1:0].
\]

We now analyze the local equation near $p_\infty$. Work in the affine chart $Z=1$, where $p_\infty$ is identified with the origin $(X,Y,W)=(0,0,0)$. By the Jacobian criterion, the homogenized system drops rank there, so $p_\infty$ is a singular point of the projective closure. On the other hand,
\[
F_1^h(X,Y,1,W)=Y^2-(Xs-Yc)^2+2W-W^2
\]
satisfies
\[
\frac{\partial F_1^h}{\partial W}(X,Y,1,W)=2-2W,
\qquad
\frac{\partial F_1^h}{\partial W}(0,0,0)=2\neq 0.
\]
Therefore the Implicit Function Theorem yields a unique local analytic solution
\[
W=\omega(X,Y)
\]
with $\omega(0,0)=0$.

To determine its leading term, observe that $\omega$ begins in degree two. Keeping only the quadratic part of the equation
\[
F_1^h(X,Y,1,\omega(X,Y))=0,
\]
the term $W^2$ contributes only in degree four and higher, so we obtain
\[
2W=(Xs-Yc)^2-Y^2+O(\|(X,Y)\|^4).
\]
Hence
\[
\omega(X,Y)=Q(X,Y)+O(\|(X,Y)\|^4),
\]
where
\[
Q(X,Y)=\tfrac12\bigl((Xs-Yc)^2-Y^2\bigr)
      =\tfrac12\bigl(s^2X^2-2sc\,XY-s^2Y^2\bigr).
\]

We next substitute this expansion into the second homogenized equation
\[
F_2^h(X,Y,1,W)
=
\bigl(X^2-2kW+(R^2+k^2)W^2\bigr)^2-4R^2W^2(X^2+Y^2).
\]
Since $W=O(\|(X,Y)\|^2)$, every term containing $W^2$ contributes only in higher order. Therefore the first nonzero homogeneous term comes from
\[
\bigl(X^2-2kW\bigr)^2
\]
with $W$ replaced by its quadratic part $Q(X,Y)$. Thus the initial form is
\[
\bigl(X^2-2kQ(X,Y)\bigr)^2=E(X,Y)^2,
\]
where
\[
E(X,Y)=X^2-2kQ(X,Y).
\]

Expanding,
\[
E(X,Y)=(1-ks^2)X^2+(2ksc)XY+(ks^2)Y^2.
\]
Its discriminant is
\begin{align*}
\Delta_Q
&=(2ksc)^2-4(1-ks^2)(ks^2)\\
&=4k^2s^2c^2-4ks^2+4k^2s^4\\
&=4k^2s^2(c^2+s^2)-4ks^2\\
&=4k^2s^2-4ks^2\\
&=4ks^2(k-1),
\end{align*}
using $c^2+s^2=1$.

If $\Delta_Q\neq 0$, then $E(X,Y)$ factors into two distinct linear forms, so the tangent cone $E(X,Y)^2=0$ consists of two distinct double directions. If $\Delta_Q=0$, those directions merge or degenerate. Thus $\Delta_Q=0$ is the algebraic condition for degeneration of the tangent-direction factorization at infinity.

Accordingly, the candidate asymptotic degeneration locus is
\[
\Sigma_\infty \subseteq \{k=0,\;k=1,\;s=0\}.
\]
Under the standing assumption $t\neq 0$, equivalently $s\neq 0$, this reduces to
\[
\Sigma_\infty \subseteq \{k=0,\;k=1\}.
\]
Hence, in the present family, the only candidate transition values contributed by infinity are $k=0$ and $k=1$.
\end{proof}

Theorem~\ref{thm:infinity-analysis} identifies the candidate transition set contributed by infinity. Together with Theorem~\ref{thm:affine-smooth}, it reduces the exact certification problem to the connected components of the complement of these candidate walls.

\subsection{Parameter chambers and exact representatives}

Since $\Sigma_{\mathrm{affine}}=\emptyset$ and the only candidate transition values at infinity are $k=0$ and $k=1$, the admissible parameter space is cut by these two walls together with the sign condition on $t$.

\begin{proposition}\label{prop:witness-cells}
Under the assumptions $R>0$ and $t\neq 0$, the boundary values $k=0$ and $k=1$, together with the sign of $t$, partition the parameter space into six open chambers. Each chamber contains a rational point.
\end{proposition}

\begin{proof}
The inequalities $t>0$ and $t<0$ split the admissible region into two open half-regions, and the two values $k=0$ and $k=1$ subdivide each of them into three open chambers corresponding to
\[
k<0,\qquad 0<k<1,\qquad k>1.
\]
Since all defining inequalities are strict, every chamber contains a rational interior point. For example, one may take
\[
(k,R,t)=(-1,1,\pm1),\qquad (1/2,1,\pm1),\qquad (2,1,\pm1).
\]
\end{proof}

These chambers are the regions on which we later perform witness-based exact certification. Before doing so, we record a simple symmetry that cuts the number of topologically distinct cases in half.

\begin{lemma}[The sign of $t$ does not change the real projective topology]
\label{lem:t-sign}
For every admissible parameter triple $(k,R,t)$ with $R>0$, $t\neq 0$, and $k\neq 0,1$, the trisectors corresponding to $(k,R,t)$ and $(k,R,-t)$ are projectively equivalent. In particular, the sign of $t$ does not affect the real projective topology of the trisector.
\end{lemma}

\begin{proof}
Under the change $t\mapsto -t$, we have $c=\cos\alpha$ unchanged and $s=\sin\alpha$ replaced by $-s$. Let
\[
\rho(x,y,z)=(-x,y,z).
\]
Then
\[
F_1^{(-s)}(x,y,z)
=
y^2-(-xs-yc)^2+2z-1
=
y^2-((-x)s-yc)^2+2z-1
=
F_1^{(s)}(\rho(x,y,z)),
\]
while $F_2$ depends only on $x^2$ and $y^2$ and is therefore invariant under $\rho$. Hence $\rho$ maps the trisector for $(k,R,t)$ to the trisector for $(k,R,-t)$. After projective closure, this becomes the linear projective transformation
\[
[X:Y:Z:W]\longmapsto [-X:Y:Z:W].
\]
Therefore the two families are projectively equivalent.
\end{proof}

For the certification pipeline, we choose one rational witness in each chamber and perform the remaining symbolic checks on the corresponding instantiated system over $\mathbb{Q}$. By Lemma~\ref{lem:t-sign}, however, the chambers with $t>0$ and $t<0$ occur in projectively equivalent pairs. Hence, up to real projective topology, it suffices to analyze the three regimes
\[
k<0,\qquad 0<k<1,\qquad k>1.
\]
In this way, the parametric transition-set problem is reduced to finitely many exact local computations, one for each topological regime in parameter space. The next section carries out this refinement in slope coordinates, allowing us to recover the actual local real structure at infinity and to prove chamberwise topological constancy.

\section{Local Branch Structure at Infinity}\label{sec:local-infinity}

We now determine the actual local real structure near the unique point at infinity. By Theorem~\ref{thm:infinity-analysis}, the only candidate transition values are $k=0$ and $k=1$. By Lemma~\ref{lem:t-sign}, the chambers with $t>0$ and $t<0$ are projectively equivalent. Hence, up to real projective topology, it is enough to analyze one representative from each of the three regimes
\[
k<0,\qquad 0<k<1,\qquad k>1.
\]

The tangent-cone computation from Section~\ref{sec:infinity} identifies the possible asymptotic directions, but by itself it does not determine the realized real branch pattern. To recover that pattern, we refine the local model in slope coordinates. We then use explicit boundary factorizations to show that the candidate walls are actual topological transition values.

\subsection{Slope-chart analysis at infinity}

We work in the affine chart $Z=1$ centered at
\[
p_\infty=[0:0:1:0].
\]
From the proof of Theorem~\ref{thm:infinity-analysis}, the equation $F_1^h=0$ can be solved locally as
\[
W=\omega(X,Y)=Q(X,Y)+O(\|(X,Y)\|^4),
\]
where
\[
Q(X,Y)=\frac12\bigl(s^2X^2-2sc\,XY-s^2Y^2\bigr).
\]

We now introduce slope coordinates
\[
Y=uX.
\]
Then
\[
Q(X,uX)=X^2q(u),
\qquad
q(u)=\frac12\bigl(s^2-2sc\,u-s^2u^2\bigr),
\]
and the quadratic tangent-direction form becomes
\[
E(X,uX)=X^2g_k(u),
\qquad
g_k(u)=1-ks^2+2ksc\,u+ks^2u^2.
\]
Its discriminant is
\[
\Delta(g_k)=4ks^2(k-1)=\Delta_Q.
\]

Substituting
\[
W=\omega(X,uX)=X^2q(u)+O(X^4)
\]
into the second homogenized equation
\[
F_2^h(X,Y,1,W)=\bigl(X^2-2kW+(R^2+k^2)W^2\bigr)^2-4R^2W^2(X^2+Y^2),
\]
we obtain
\begin{align*}
F_2^h\bigl(X,uX,1,\omega(X,uX)\bigr)
&=
\Bigl(X^2-2kX^2q(u)+(R^2+k^2)X^4q(u)^2+O(X^6)\Bigr)^2 \\
&\quad -4R^2X^6q(u)^2(1+u^2)+O(X^8) \\
&=
X^4g_k(u)^2 \\
&\quad +X^6\Bigl(2(R^2+k^2)g_k(u)q(u)^2-4R^2(1+u^2)q(u)^2\Bigr)
+O(X^8).
\end{align*}
Dividing by $X^4$, we obtain the reduced local equation
\begin{equation}\label{eq:reduced-slope-main}
G(X,u)=g_k(u)^2+X^2H(X,u)=0,
\end{equation}
where
\[
H(X,u)=2(R^2+k^2)g_k(u)q(u)^2-4R^2(1+u^2)q(u)^2+O(X^2).
\]

The crucial sign comes from evaluating the coefficient of $X^2$ at a real root of $g_k$.

\begin{lemma}\label{lem:sign-H-main}
If $u_0$ is a real root of $g_k$, then $q(u_0)\neq 0$ and
\[
H(0,u_0)=-4R^2(1+u_0^2)q(u_0)^2<0.
\]
\end{lemma}

\begin{proof}
If $q(u_0)=0$, then
\[
g_k(u_0)=1-2kq(u_0)=1,
\]
contrary to the assumption that $u_0$ is a root of $g_k$. Hence $q(u_0)\neq 0$. Evaluating the coefficient of $X^2$ in \eqref{eq:reduced-slope-main} at $u=u_0$ gives
\[
H(0,u_0)
=
2(R^2+k^2)g_k(u_0)q(u_0)^2-4R^2(1+u_0^2)q(u_0)^2
=
-4R^2(1+u_0^2)q(u_0)^2<0.
\]
\end{proof}

The sign of $\Delta_Q$ now determines the number of real slope points, and Lemma~\ref{lem:sign-H-main} determines the local branch structure at each such point.

\begin{theorem}[Chamberwise local constancy at infinity]
\label{thm:local-infinity}
For the canonical symmetric line--line--circle family, the local real projective structure near $p_\infty$ is constant on each connected component of
\[
\{R>0,\ t\neq 0,\ k\neq 0,\ k\neq 1\}.
\]
More precisely:
\begin{enumerate}
    \item If $k<0$ or $k>1$, then $g_k$ has two distinct real roots $u_+$ and $u_-$, and near each point $(0,u_\pm)$ the reduced slope-chart curve has a transverse real $X$-shaped local crossing.
    \item If $0<k<1$, then $g_k$ has no real root, and no real branch approaches the exceptional line $X=0$ at finite slope.
\end{enumerate}
Consequently, the local topology at infinity is unchanged within each chamber of the complement of $\{k=0\}\cup\{k=1\}$.
\end{theorem}

\begin{proof}
Assume first that $k<0$ or $k>1$. Then
\[
\Delta_Q=4ks^2(k-1)>0,
\]
so the quadratic polynomial $g_k$ has two distinct real roots $u_+$ and $u_-$. Since these roots are simple, we have
\[
g_k(u)=g_k'(u_\pm)(u-u_\pm)+O((u-u_\pm)^2)
\]
near each root. Substituting this expansion into \eqref{eq:reduced-slope-main} and using Lemma~\ref{lem:sign-H-main}, we obtain
\[
g_k'(u_\pm)^2(u-u_\pm)^2
=
4R^2(1+u_\pm^2)q(u_\pm)^2X^2+O(X^3).
\]
Hence
\[
(u-u_\pm)^2=\lambda_\pm X^2+O(X^3),
\qquad
\lambda_\pm=
\frac{4R^2(1+u_\pm^2)q(u_\pm)^2}{g_k'(u_\pm)^2}>0.
\]
This is precisely the local equation of a transverse real $X$-shaped crossing in the $(X,u)$-plane.

Now assume that $0<k<1$. Then
\[
\Delta_Q<0,
\]
so $g_k$ has no real root. Since the leading coefficient of $g_k$ is $ks^2>0$, the quadratic polynomial $g_k$ has positive minimum on $\mathbb{R}$; in particular,
\[
g_k(u)^2\ge m>0
\]
for some constant $m$. Therefore, for sufficiently small $|X|$, the perturbation term $X^2H(X,u)$ cannot cancel the strictly positive leading term in \eqref{eq:reduced-slope-main} at any finite real slope $u$. Thus the reduced curve has no real point over the exceptional line $X=0$, and no real branch approaches $p_\infty$ at finite slope.

In both cases, the local branch pattern is determined entirely by the chamber containing the parameter value. Therefore the local topology at infinity remains constant on each connected component of the complement of $\{k=0\}\cup\{k=1\}$.
\end{proof}

Theorem~\ref{thm:local-infinity} shows that the values $k=0$ and $k=1$ are the only places where the local branch structure at infinity can change. To conclude the transition-set computation, it remains to show that these candidate walls are actual topological transition values.

\subsection{The degenerate boundary values $k=1$ and $k=0$}

We now examine the two boundary values at which the tangent-cone discriminant vanishes. The point is not merely that the tangent directions degenerate there, but that the trisector itself undergoes an explicit algebraic transition in affine space.

\subsubsection{The boundary value $k=1$}

Geometrically, the condition $k=1$ means that the circle $C$ lies in the same distinguished plane as the second line $L_2$, namely the plane $z=1$. In this regime, the line--circle bisector becomes reducible.

\begin{theorem}\label{thm:degenerate-k1}
At the boundary value $k=1$, the degree-eight trisector becomes reducible and splits into two algebraic space curves of degree four. These two quartic components intersect in exactly two real affine points:
\[
\left(Rc,Rs,\frac{1-R^2s^2}{2}\right)
\qquad\text{and}\qquad
\left(-Rc,-Rs,\frac{1-R^2s^2}{2}\right).
\]
Moreover, each of these two affine intersection points is an ordinary real node of the boundary trisector.
\end{theorem}

\begin{proof}
Substituting $k=1$ into the line--circle bisector gives
\[
F_2(x,y,z)=\bigl(x^2-2z+R^2+1\bigr)^2-4R^2(x^2+y^2)=0.
\]
From the line--line bisector
\[
F_1=y^2-(xs-yc)^2+2z-1=0,
\]
we solve for $z$:
\[
z=\frac12(xs-yc)^2-\frac12y^2+\frac12.
\]
Substituting this into the first term of $F_2$ yields
\[
x^2-2z+R^2+1=x^2+y^2-(xs-yc)^2+R^2.
\]
Using the rotational identity
\[
x^2+y^2=(xc+ys)^2+(xs-yc)^2,
\]
we obtain the projected equation
\[
P(x,y)=\bigl((xc+ys)^2+R^2\bigr)^2-4R^2\bigl((xc+ys)^2+(xs-yc)^2\bigr)=0.
\]
Rearranging,
\[
P(x,y)=\bigl((xc+ys)^2-R^2\bigr)^2-\bigl(2R(xs-yc)\bigr)^2.
\]
Hence
\[
P(x,y)=P_1(x,y)\,P_2(x,y),
\]
where
\[
P_1(x,y)=(xc+ys)^2-2R(xs-yc)-R^2,
\qquad
P_2(x,y)=(xc+ys)^2+2R(xs-yc)-R^2.
\]
Thus, at $k=1$, the quartic surface $F_2=0$ decomposes into the union of two quadratic cylinders. Intersecting each of these degree-two surfaces with the degree-two surface $F_1=0$ yields a degree-four space curve. Therefore the original degree-eight trisector splits into two quartic components.

To compute the intersection of these two quartic curves, it suffices to solve
\[
P_1(x,y)=0,\qquad P_2(x,y)=0.
\]
Subtracting the two equations gives
\[
4R(xs-yc)=0,
\]
hence
\[
xs=yc.
\]
Adding the two equations gives
\[
2(xc+ys)^2-2R^2=0,
\]
so
\[
(xc+ys)^2=R^2.
\]
Using $xs=yc$, we obtain
\[
(x,y)=(Rc,Rs)
\qquad\text{or}\qquad
(x,y)=(-Rc,-Rs).
\]
Substituting either point into $F_1$ gives
\[
z=\frac{1-R^2s^2}{2}.
\]
Therefore the two quartic components intersect in exactly the two real affine points claimed in the statement.

It remains to verify that these are ordinary real nodes. Let
\[
C_1=\{F_1=0,\ P_1=0\},
\qquad
C_2=\{F_1=0,\ P_2=0\}.
\]
At a common point, each component is smooth if the gradients of its defining surfaces are linearly independent, and the two components meet transversely if their tangent lines are distinct.

We compute
\[
\nabla F_1(x,y,z)=\bigl(-2s(xs-yc),\ 2y+2c(xs-yc),\ 2\bigr),
\]
\[
\nabla P_1(x,y,z)=2(c,s,0)(xc+ys)-2R(s,-c,0),
\qquad
\nabla P_2(x,y,z)=2(c,s,0)(xc+ys)+2R(s,-c,0).
\]
At
\[
p_+=\left(Rc,Rs,\frac{1-R^2s^2}{2}\right),
\]
we have $xc+ys=R$ and $xs-yc=0$, hence
\[
\nabla F_1(p_+)=(0,\,2Rs,\,2),
\]
\[
\nabla P_1(p_+)=2R(c-s,\ s+c,\ 0),
\qquad
\nabla P_2(p_+)=2R(c+s,\ s-c,\ 0).
\]
The determinant of the $3\times 3$ matrix formed by these three normals is
\[
\det
\begin{pmatrix}
0 & 2Rs & 2\\
2R(c-s) & 2R(s+c) & 0\\
2R(c+s) & 2R(s-c) & 0
\end{pmatrix}
=-16R^2\neq 0.
\]
Thus $C_1$ and $C_2$ are both smooth at $p_+$ and have distinct tangent lines there. The same computation at
\[
p_-=\left(-Rc,-Rs,\frac{1-R^2s^2}{2}\right)
\]
again yields a nonzero determinant up to sign. Hence both affine intersection points are ordinary real nodes.
\end{proof}

Theorem~\ref{thm:degenerate-k1} shows that the candidate boundary $k=1$ is accompanied by an explicit algebraic transition: the irreducible degree-eight trisector of the generic region becomes reducible and splits into two quartic components meeting transversely in affine space.

\subsubsection{The boundary value $k=0$}

We next examine the second boundary value $k=0$. This is the other parameter value at which the tangent-cone discriminant $\Delta_Q=4ks^2(k-1)$ vanishes. As in the case $k=1$, the point is not merely that the tangent-direction pattern degenerates, but that the affine structure of the trisector undergoes a genuine algebraic transition.

\begin{theorem}\label{thm:degenerate-k0}
At the boundary value $k=0$, the degree-eight trisector becomes reducible and splits into two algebraic space curves of degree four. These two quartic components intersect in exactly two real affine points:
\[
\left(R,0,\frac{1+R^2s^2}{2}\right)
\qquad\text{and}\qquad
\left(-R,0,\frac{1+R^2s^2}{2}\right).
\]
Moreover, each of these two affine intersection points is an ordinary real node of the boundary trisector.
\end{theorem}

\begin{proof}
Substituting $k=0$ into the line--circle bisector gives
\[
F_2(x,y,z)=\bigl(x^2+R^2\bigr)^2-4R^2(x^2+y^2)=0.
\]
Rearranging,
\[
F_2=(x^2-R^2)^2-(2Ry)^2.
\]
Hence
\[
F_2(x,y,z)=P_+(x,y)\,P_-(x,y),
\]
where
\[
P_+(x,y)=x^2+2Ry-R^2,
\qquad
P_-(x,y)=x^2-2Ry-R^2.
\]
Thus, at $k=0$, the quartic surface $F_2=0$ decomposes into the union of two quadratic cylinders. Intersecting each of these degree-two surfaces with the degree-two surface
\[
F_1=y^2-(xs-yc)^2+2z-1=0
\]
yields a degree-four space curve. Therefore the original degree-eight trisector splits into two quartic components
\[
C_+=\{F_1=0,\ P_+=0\},
\qquad
C_-=\{F_1=0,\ P_-=0\}.
\]

To compute their intersection, it suffices to solve
\[
P_+(x,y)=0,\qquad P_-(x,y)=0.
\]
Subtracting the two equations gives
\[
4Ry=0.
\]
Since $R>0$, it follows that $y=0$. Substituting into either equation yields
\[
x^2-R^2=0,
\]
so
\[
x=\pm R.
\]
Substituting either point into $F_1$ gives
\[
-(xs)^2+2z-1=0,
\]
hence
\[
z=\frac{1+R^2s^2}{2}.
\]
Therefore the two quartic components intersect in exactly the two real affine points claimed in the statement.

We now verify that these are ordinary real nodes. We compute
\[
\nabla F_1(x,y,z)=\bigl(-2s(xs-yc),\ 2y+2c(xs-yc),\ 2\bigr),
\]
\[
\nabla P_+(x,y,z)=(2x,\,2R,\,0),
\qquad
\nabla P_-(x,y,z)=(2x,\,-2R,\,0).
\]
At
\[
q_+=\left(R,0,\frac{1+R^2s^2}{2}\right),
\]
we obtain
\[
\nabla F_1(q_+)=(-2Rs^2,\,2Rsc,\,2),
\]
\[
\nabla P_+(q_+)=(2R,\,2R,\,0),
\qquad
\nabla P_-(q_+)=(2R,\,-2R,\,0).
\]
Therefore
\[
\det
\begin{pmatrix}
-2Rs^2 & 2Rsc & 2\\
2R & 2R & 0\\
2R & -2R & 0
\end{pmatrix}
=-16R^2\neq 0.
\]
So $C_+$ and $C_-$ are both smooth at $q_+$ and have distinct tangent lines there. The same computation at
\[
q_-=\left(-R,0,\frac{1+R^2s^2}{2}\right)
\]
again gives a nonzero determinant up to sign. Hence both affine intersection points are ordinary real nodes.
\end{proof}

Theorem~\ref{thm:degenerate-k0} shows that the candidate boundary $k=0$ is likewise accompanied by an explicit algebraic transition: the degree-eight trisector becomes reducible and splits into two quartic components meeting transversely in affine space.

\begin{corollary}[Actual transition walls]\label{cor:actual-walls}
For the canonical symmetric line--line--circle family, the boundary values $k=0$ and $k=1$ are actual topological transition values. Equivalently, the certified topological transition set is
\[
\Sigma=\{k=0\}\cup\{k=1\}.
\]
\end{corollary}

\begin{proof}
Let $T_\lambda$ denote the trisector corresponding to an admissible parameter value $\lambda$. By Theorem~\ref{thm:affine-smooth}, every affine point of $T_\lambda$ is locally smooth. Hence, for sufficiently small $\varepsilon>0$, the punctured neighborhood of such a point has exactly two connected components.

At the boundary values $k=0$ and $k=1$, Theorems~\ref{thm:degenerate-k1} and~\ref{thm:degenerate-k0} show that the trisector becomes reducible and has ordinary real affine nodes. At such a node, a sufficiently small punctured neighborhood has four connected components.

A local ambient homeomorphism preserves the number of connected components of a punctured neighborhood. Therefore a nodal boundary member cannot be locally homeomorphic to a nearby admissible smooth member. It follows that the topology is not locally constant across either boundary. Hence both $k=0$ and $k=1$ are actual topological transition values.
\end{proof}

\section{Classical Three-Line Benchmark}\label{sec:three-line}

We conclude the main body of the paper by testing the same verification pipeline on the classical trisector of three skew lines in $\mathbb{R}^3$. The purpose of this section is not to rederive the full structural theory of Everett et al.~\cite{everett2007Voronoi}, which relies on the pencil-of-quadrics geometry specific to the all-line case. Rather, the goal is methodological: we use the three-line configuration as a control benchmark to check that the exact transition-set framework developed in this paper behaves as expected in the classical quadric-only regime.

This comparison is useful for two reasons. First, it separates genuinely new mixed-object phenomena from artifacts of the symbolic pipeline. In the three-line setting, all relevant bisectors are quadrics, so the trisector remains inside the classical pencil-of-quadrics framework. In the canonical line--line--circle family, by contrast, one bisector is quartic, and the projective analysis at infinity produces an additional transition mechanism governed by the discriminant $\Delta_Q=4ks^2(k-1)$. Thus the all-line case provides a natural baseline against which the mixed-object behavior can be measured.

Second, the benchmark clarifies the scope of our method. The framework introduced in this paper is not tied to the line--line--circle geometry itself; it is meant to be a general exact-verification template for bounded-dimensional parametric families of trisectors. Applying it to the classical three-line case shows that, when the underlying geometry remains in the quadric regime, the same pipeline recovers the expected absence of an additional asymptotic wall.

\subsection{Generic rational benchmark configuration}

Following the parameterization used by Everett et al.~\cite{everett2007Voronoi}, consider the three skew lines
\[
L_1:\ y=ax,\ z=1,
\qquad
L_2:\ y=-ax,\ z=-1,
\]
and
\[
L_3:\ (x,y,z)=(u,v,0)+\lambda(\alpha,\beta,1),
\]
where the parameters are chosen so that the three lines are pairwise skew and not all parallel to a common plane. To obtain an exact benchmark over $\mathbb{Q}$, we fix the rational witness
\[
a=2,\qquad u=1,\qquad v=2,\qquad \alpha=1,\qquad \beta=1.
\]

For this witness, the homogenized bisector equations may be written as
\[
F_{12}^h = ZW(1+a^2)+aXY
\]
and
\begin{align*}
F_{13}^h
&=
(u^2+v^2)X^2+(\alpha^2+v^2)Y^2-(u^2+\alpha^2)Z^2
-2u\alpha\,XY-2v\alpha\,XZ-2uv\,YZ \\
&\qquad
+2\beta(u^2+v^2)XW-2a\beta(u^2+v^2)YW.
\end{align*}
After substitution of the witness values, these define a concrete exact test instance over $\mathbb{Q}$.

Because both defining bisectors are quadrics, the trisector is the intersection of two quadratic surfaces, and hence lies in a degree-four projective class. The exact symbolic checks below confirm that the witness behaves exactly as expected in the generic three-line case.

\begin{theorem}\label{thm:three-line-benchmark}
For a generic three-line configuration in $\mathbb{R}^3$, the same witness-based exact pipeline used in this paper certifies that the projective trisector has degree $4$, that its affine singular locus is empty, and that no additional transition wall appears at infinity.
Equivalently, in the classical three-line benchmark there is no asymptotic contribution analogous to $\Sigma_\infty$.
\end{theorem}

\begin{proof}
We run the same three checks as in the mixed-object case: affine singularity testing, projective closure, and infinity analysis.

\paragraph{Affine part.}
Since all three sites are lines, the trisector is already defined by polynomial equations, with no radical term analogous to the line--circle distance. We form the homogenized projective ideal and saturate with respect to $W$ to obtain the true projective closure
\[
I_{\mathrm{proj}}=\operatorname{saturate}(\langle F_{12}^h,F_{13}^h\rangle,\langle W\rangle).
\]
For the generic witness above, a direct exact computation shows that the corresponding projective curve has degree $4$.

To test affine smoothness, we adjoin the Jacobian rank-drop conditions in projective space,
\[
\mathrm{SingProj}=I_{\mathrm{proj}}+\mathrm{minors}(2,\mathrm{jacobian}(I_{\mathrm{proj}})),
\]
and then remove the hyperplane at infinity:
\[
\mathrm{SingAffine}=\operatorname{saturate}(\mathrm{SingProj},\langle W\rangle).
\]
The resulting ideal is the unit ideal,
\[
\mathrm{SingAffine}=\langle 1\rangle.
\]
Hence the affine singular locus is empty.

\paragraph{Behavior at infinity.}
We now inspect the leading forms on the hyperplane at infinity $W=0$. The point is not to reproduce the full classical theory, but to explain why the exact infinity check used in the mixed-object case produces no additional wall here.

At infinity, only the highest-degree part of the squared-distance expressions matters. Since translations affect only lower-order terms, the asymptotic squared distance from
\[
P=[X:Y:Z:0]
\]
to a line with direction vector $\vec d$ is proportional to
\[
\frac{\|P\times \vec d\|^2}{\|\vec d\|^2}.
\]
For the direction vectors
\[
\vec d_1=(1,a,0),\qquad
\vec d_2=(1,-a,0),\qquad
\vec d_3=(\alpha,\beta,1),
\]
the corresponding leading forms are
\begin{align*}
D_\infty^2(P,L_1)
&=
\frac{a^2X^2-2aXY+Y^2+(1+a^2)Z^2}{1+a^2}, \\
D_\infty^2(P,L_2)
&=
\frac{a^2X^2+2aXY+Y^2+(1+a^2)Z^2}{1+a^2}, \\
D_\infty^2(P,L_3)
&=
\frac{(Y-\beta Z)^2+(\alpha Z-X)^2+(\beta X-\alpha Y)^2}
{\alpha^2+\beta^2+1}.
\end{align*}

Equating the first two expressions gives the asymptotic bisector equation
\[
F_{12}^{\infty}(X,Y,Z)=-4aXY=0.
\]
For nonparallel skew lines, $a\neq 0$, so this reduces to
\[
X=0
\qquad\text{or}\qquad
Y=0.
\]
Thus the asymptotic directions are already constrained by the line directions alone.

Substituting these two cases into the second asymptotic equation
\[
D_\infty^2(P,L_1)=D_\infty^2(P,L_3)
\]
gives explicit quadratic equations.

If $X=0$, clearing denominators yields a quadratic form
\[
AY^2+BYZ+CZ^2=0
\]
with coefficients
\[
A=\beta^2-a^2(1+\alpha^2),\qquad
B=2\beta(1+a^2),\qquad
C=1+a^2,
\]
whose discriminant is
\[
\Delta_X=4a^2(1+a^2)(\alpha^2+\beta^2+1)>0.
\]

If $Y=0$, we obtain a quadratic form in $(X,Z)$ with coefficients
\[
A=a^2\alpha^2-(1+\beta^2),\qquad
B=2\alpha(1+a^2),\qquad
C=1+a^2,
\]
whose discriminant is
\[
\Delta_Y=4(1+a^2)(\alpha^2+\beta^2+1)>0.
\]

Hence the three-line trisector has four distinct real asymptotic directions, and the exact infinity check produces no additional parameter discriminant or transition wall analogous to the mixed-object factor $\Delta_Q=4ks^2(k-1)$.

Therefore the witness-based transition-set pipeline recovers exactly the expected classical behavior: quartic degree, affine smoothness, and no extra asymptotic wall.
\end{proof}

Theorem~\ref{thm:three-line-benchmark} should be interpreted as a control benchmark rather than a replacement for the classical pencil-of-quadrics theory. The point is that the exact verification framework used throughout this paper does not create spurious walls at infinity: when applied to the all-line case, it recovers the expected absence of an additional asymptotic transition.

\subsection{Compatibility with the degenerate cubic-plus-line case}

For completeness, we also record a symbolic benchmark for the classical degenerate case in which the trisector splits into a nonsingular cubic and a line. This is the standard nongeneric degeneration described by Everett et al.~\cite{everett2007Voronoi}, and it provides a second consistency check for the present framework.

Everett et al.~\cite{everett2007Voronoi} showed that the quartic trisector degenerates into a nonsingular cubic and a line precisely when the hyperboloid containing the three lines is a hyperboloid of revolution, and that these two components do not meet in the real projective space $\mathbb{P}^3(\mathbb{R})$. Their proof uses the characteristic polynomial of the associated pencil of quadrics. Our goal here is not to reprove that theory, but to verify that the same symbolic viewpoint used in the mixed-object case is compatible with this standard degenerate benchmark.

Applying the corresponding degeneracy condition to the present parameterization yields a rational witness configuration, for example
\[
a=2,\qquad
(u,v)=(10,-4),\qquad
(\alpha,\beta)=(2,0).
\]
Substituting these values into the bisector equations produces a concrete exact test instance over $\mathbb{Q}$.

For this witness, the saturated projective closure decomposes into two irreducible components of degrees $3$ and $1$, respectively. Thus the trisector splits as a cubic together with a line. Moreover, the cubic component is nonsingular, and the two components have no real intersection points. Algebraically, their projective intersection is zero-dimensional over $\mathbb{C}$, but all such points are nonreal: there is no real projective intersection at infinity, and after dehomogenization to affine space the elimination polynomial reduces to
\[
Z^2+125=0,
\]
which has no real root.

We summarize this compatibility check as follows.

\begin{proposition}\label{prop:three-line-degenerate}
The witness-based symbolic framework remains consistent even in the classical degenerate three-line case: for a rational witness in this regime, the projective trisector correctly splits into a nonsingular cubic and a line, with no real intersection points in $\mathbb{P}^3(\mathbb{R})$.
\end{proposition}

\begin{proof}
For the witness
\[
a=2,\qquad (u,v)=(10,-4),\qquad (\alpha,\beta)=(2,0),
\]
the corresponding bisector equations define a saturated projective ideal whose primary decomposition has two components of degrees $3$ and $1$. This gives the cubic-line splitting.

To check the geometry of the cubic component, one computes its projective singular locus and finds that it is empty, so the cubic is nonsingular. To check the intersection with the line, one forms the sum of the two component ideals. The resulting intersection ideal is zero-dimensional over $\mathbb{C}$, but after dehomogenizing to affine space and eliminating $X$ and $Y$, one obtains the polynomial
\[
Z^2+125.
\]
Since this polynomial has no real root, the cubic and the line have no real affine intersection. In addition, there is no real projective intersection on the hyperplane at infinity. Therefore the two components do not meet in $\mathbb{P}^3(\mathbb{R})$.
\end{proof}

The consistency across both generic and degenerate three-line cases demonstrates that our symbolic framework faithfully recovers the classical trisector structure. This reinforces the status of the line--line--circle transition walls as an intrinsic non-quadric phenomenon: the additional wall identified in the mixed-object family is not a byproduct of the verification pipeline, but a fundamental geometric feature that emerges only when transcending the quadric-only regime.

For reproducibility, the raw verification scripts for these benchmark computations are included in the appendix.
The mathematical content, however, is contained entirely in the main text.

\section{Cost of the Verification Pipeline}\label{sec:complexity}

We briefly record the symbolic cost of the verification pipeline used in this paper. The purpose of this section is not to claim a competitive alternative to full CAD in arbitrary dimension, nor to optimize constant factors in a practical implementation. Rather, the point is conceptual: for the bounded-dimensional parametric families considered here, the verification task can be organized so that the dependence on the parameter dimension is controlled and separated from the fixed-dimensional geometry of the trisector itself.

Recall that the pipeline has two logically distinct stages. First, we construct the candidate transition set
\[
\Sigma=\Sigma_{\mathrm{affine}}\cup\Sigma_{\infty},
\]
where $\Sigma_{\mathrm{affine}}$ comes from affine Jacobian rank-drop conditions and $\Sigma_{\infty}$ comes from degeneration of the tangent-direction equation at infinity. Second, we choose one rational witness from each connected component of $U\setminus\Sigma$ and run exact symbolic checks on the corresponding instantiated system.

The important structural point is that the trisector itself is \emph{not} reduced to finitely many points. We do not attempt a full algebraic decomposition of the entire curve. Instead, for a fixed witness parameter, only the singularity test is reduced to a zero-dimensional ideal, after adjoining Jacobian minors and removing components supported at infinity. This is what keeps the geometric part of the computation under control in fixed ambient dimension.

\begin{proposition}[Symbolic cost in the bounded-dimensional regime]\label{prop:complexity}
Assume that the ambient spatial dimension is fixed, say $n=3$, and that the number of defining geometric equations is bounded by a constant $s$. Under standard elimination bounds, the verification pipeline admits the estimate
\[
O\!\left(s^{m+1}d^{O(m)} + s^m d^{O(n)}\right),
\]
where $m$ is the parameter dimension and $d$ is the maximum polynomial degree.

The first term corresponds to constructing the candidate transition set in parameter space and extracting one witness from each connected component of its complement. The second term corresponds to witness-based exact chamber checks on fixed-dimensional instantiated systems. In particular, for fixed $n$ and bounded $s$, the dependence on the parameter dimension is singly exponential.
\end{proposition}

\begin{proof}
We separate the verification pipeline into two parts: construction of the candidate transition set in parameter space, and witness-based exact checks on one sampled parameter value from each connected chamber.

\paragraph{Part 1: constructing the candidate transition set.}
We compute
\[
\Sigma=\Sigma_{\mathrm{affine}}\cup\Sigma_{\infty},
\]
where $\Sigma_{\mathrm{affine}}$ is defined by affine Jacobian rank-drop conditions and $\Sigma_{\infty}$ by degeneration of the direction equation at infinity. The key point is that these elimination steps are carried out only over the geometric variables together with a bounded number of auxiliary variables. Since the ambient spatial dimension is fixed in the present paper, standard elimination bounds give a symbolic cost of the form
\[
d^{O(n)}
\]
for each such fixed-dimensional elimination step.

The result is a semi-algebraic description of $\Sigma$ in the $m$-dimensional parameter space, involving finitely many polynomial constraints of degree bounded by $d^{O(1)}$. Once $\Sigma$ has been constructed, the complement $U\setminus\Sigma$ is a semi-algebraic set in parameter space. By the Milnor--Thom bound, the number of its connected components satisfies
\[
\beta_0(U\setminus\Sigma)=O((sd)^m).
\]
Extracting one rational witness point from each component can be handled by standard decision procedures for existential formulas over the reals, and this yields the singly-exponential parameter-space contribution
\[
O\!\left(s^{m+1}d^{O(m)}\right).
\]

\paragraph{Part 2: witness extraction and exact chamber checks.}
Fix one witness parameter $\lambda_j$. After instantiation, the trisector remains a one-dimensional algebraic curve in fixed ambient dimension. We do \emph{not} reduce the curve itself to finitely many points. Rather, only the singularity test is reduced to a zero-dimensional problem: we adjoin the Jacobian-minor conditions and remove components supported at infinity.

Because this instantiated singularity computation takes place in fixed geometric dimension $n$, standard Bézout-type degree bounds imply that the corresponding quotient algebra has size at most
\[
d^{O(n)}.
\]
Thus each exact chamber check costs
\[
d^{O(n)}.
\]
Since the number of witness chambers is $O((sd)^m)$, the total cost of the chamberwise stage is
\[
O\!\left((sd)^m d^{O(n)}\right).
\]
Using the same coarse asymptotic notation as above, we record this as
\[
O\!\left(s^m d^{O(n)}\right).
\]

Combining the two parts yields
\[
O\!\left(s^{m+1}d^{O(m)} + s^m d^{O(n)}\right).
\]
This proves the stated estimate.
\end{proof}

The meaning of Proposition~\ref{prop:complexity} is qualitative rather than absolute. The estimate shows that the hard symbolic work is concentrated in parameter-space transition computation, while the chamberwise certification stage is performed on fixed-dimensional exact instances. This is precisely the sense in which the present framework avoids the full cost of an ambient CAD decomposition: instead of decomposing the entire geometric space, it isolates a transition set in parameter space and then certifies one representative from each stable chamber.

Finally, we emphasize that this section records only a worst-case symbolic bound for the abstract verification pipeline. Our actual implementation uses Z3~\cite{de2008z3} for witness extraction and Macaulay2~\cite{M2} for exact symbolic computation, both of which rely on practical heuristics and may behave quite differently from the worst-case model.

\section{Conclusion}\label{sec:conc}
We studied exact verification for bounded-dimensional families of algebraic curves arising from Voronoi diagrams, and proposed a witness-based symbolic pipeline for this setting. The pipeline separates parameter-space search from exact local checking: it first computes candidate transition walls, then samples one rational witness from each parameter region, and finally performs exact checks on the corresponding instantiated system.

As a first non-quadric bisector case, we analyzed a canonical symmetric line--line--circle trisector family in $\mathbb{R}^3$. Even in this restricted family, one bisector is quartic, the trisector has degree eight, and the behavior at infinity becomes essential. For this family, we proved affine smoothness, identified a unique point at infinity, derived the explicit direction discriminant governing the candidate walls, and used slope coordinates to determine the local branch structure at infinity. This yields the certified transition set $\Sigma=\{k=0\}\cup\{k=1\}$.

Thus, the paper should be read as a first exact mixed-object base case for witness-based symbolic verification in Voronoi diagrams. It is not a full topological classification of arbitrary mixed-object trisectors, nor a replacement for CAD in arbitrary dimension. Rather, it shows that a meaningful bounded-dimensional class can already be handled exactly by combining witness sampling with local symbolic verification.

Several natural questions remain open. First, our transition-set computation leaves open the global comparison between chambers: do the outer regimes $k<0$ and $k>1$ give the same real projective topology, or do they represent genuinely different curve types? Second, the next mixed-object base cases are the trisectors of type line--circle--circle and circle--circle--circle. Third, once such three-site predicates are understood, the next challenge is to analyze four-site mixed-object quadrisectors, with the long-term goal of exact predicates for Voronoi vertices and more general mixed-object Voronoi diagrams in $\mathbb{R}^3$.

\section*{Acknowledgement}
The authors would like to thank the members of BRL AGSTA for their extensive support throughout this research. In particular, we are deeply grateful to Jaewoo Jung and Maciej Ga\l{}azka for their insightful discussions and comments on bridging the gap between algebraic geometry and computational geometry.

\pagebreak

\small
\newpage
\bibliographystyle{plainurl}
\bibliography{Voronoi}

@article{aurenhammer1991voronoi,
  title={Voronoi diagrams—a survey of a fundamental geometric data structure},
  author={Aurenhammer, Franz},
  journal={ACM computing surveys (CSUR)},
  volume={23},
  number={3},
  pages={345--405},
  year={1991},
  publisher={ACM New York, NY, USA}
}

@book{edelsbrunner1987algorithms,
  title={Algorithms in combinatorial geometry},
  author={Edelsbrunner, Herbert},
  volume={10},
  year={1987},
  publisher={Springer Science \& Business Media}
}

@article{chazelle1996lines,
  title={Lines in space: Combinatorics and algorithms},
  author={Chazelle, Bernard and Edelsbrunner, Herbert and Guibas, Leonidas J. and Sharir, Micha and Stolfi, Jorge},
  journal={Algorithmica},
  volume={15},
  number={5},
  pages={428--447},
  year={1996},
  publisher={Springer}
}

@article{sharir1994almost,
  title={Almost tight upper bounds for lower envelopes in higher dimensions},
  author={Sharir, Micha},
  journal={Discrete \& Computational Geometry},
  volume={12},
  number={3},
  pages={327--345},
  year={1994},
  publisher={Springer}
}

@inproceedings{everett2007voronoi,
  title={The Voronoi diagram of three lines},
  author={Everett, Hazel and Lazard, Sylvain and Lazard, Daniel and El Din, Mohab Safey},
  booktitle={Proceedings of the twenty-third annual symposium on Computational geometry},
  pages={255--264},
  year={2007}
}

@article{davenport1988real,
  title={Real quantifier elimination is doubly exponential},
  author={Davenport, James H and Heintz, Joos},
  journal={Journal of Symbolic Computation},
  volume={5},
  number={1-2},
  pages={29--35},
  year={1988},
  publisher={Elsevier}
}

@book{cox1997ideals,
  title={Ideals, varieties, and algorithms},
  author={Cox, David and Little, John and O'shea, Donal and Sweedler, Moss},
  year={1997},
  publisher={Springer}
}

@book{basu2006algorithms,
  title={Algorithms in real algebraic geometry},
  author={Basu, Saugata and Pollack, Richard and Roy, Marie-Fran{\c{c}}oise},
  year={2006},
  publisher={Springer}
}

@phdthesis{anton2004voronoi,
  title={Voronoi diagrams of semi-algebraic sets},
  author={Anton, Fran{\c{c}}ois},
  year={2004},
  school={University of British Columbia}
}

@inproceedings{hemmer2010constructing,
  title={Constructing the Exact Voronoi Diagram of Arbitrary Lines in Three-Dimensional Space: with Fast Point-Location},
  author={Hemmer, Michael and Setter, Ophir and Halperin, Dan},
  booktitle={European Symposium on Algorithms},
  pages={398--409},
  year={2010},
  organization={Springer}
}

@inproceedings{berberich2011generic,
  title={A generic algebraic kernel for non-linear geometric applications},
  author={Berberich, Eric and Hemmer, Michael and Kerber, Michael},
  booktitle={Proceedings of the twenty-seventh annual symposium on Computational geometry},
  pages={179--186},
  year={2011}
}

@article{schwartz1983piano,
  title={On the “piano movers” problem. II. General techniques for computing topological properties of real algebraic manifolds},
  author={Schwartz, Jacob T and Sharir, Micha},
  journal={Advances in applied Mathematics},
  volume={4},
  number={3},
  pages={298--351},
  year={1983},
  publisher={Elsevier}
}

@inproceedings{hanniel2008computing,
  title={Computing the Voronoi cells of planes, spheres and cylinders in R3},
  author={Hanniel, Iddo and Elber, Gershon},
  booktitle={Proceedings of the 2008 ACM symposium on Solid and physical modeling},
  pages={47--58},
  year={2008}
}

@article{adamou2021computing,
  title={Computing the Topology of Vorono{\"\i} Diagrams of Parallel Half-Lines},
  author={Adamou, Ibrahim and Mourrain, Bernard},
  journal={Mathematics in Computer Science},
  volume={15},
  number={4},
  pages={859--876},
  year={2021},
  publisher={Springer}
}

@article{fulton1969algebraic,
  title={Algebraic curves: an introduction to algebraic geometry},
  author={Fulton, William and Weiss, Richard},
  journal={(No Title)},
  year={1969}
}

@Misc{M2,
  author = {Grayson, Daniel R. and Stillman, Michael E.},
  title = {Macaulay2, a software system for research in algebraic geometry},
  howpublished = {Available at \url{http://www2.macaulay2.com}}
}

@inproceedings{de2008z3,
  title={Z3: An efficient SMT solver},
  author={De Moura, Leonardo and Bj{\o}rner, Nikolaj},
  booktitle={International conference on Tools and Algorithms for the Construction and Analysis of Systems},
  pages={337--340},
  year={2008},
  organization={Springer}
}

@inproceedings{alt1995voronoi,
  title={The Voronoi diagram of curved objects},
  author={Alt, Helmut and Schwarzkopf, Otfried},
  booktitle={Proceedings of the eleventh annual symposium on Computational geometry},
  pages={89--97},
  year={1995}
}

@article{barequet2024unbounded,
  title={Unbounded regions of high-order Voronoi diagrams of lines and line segments in higher dimensions},
  author={Barequet, Gill and Papadopoulou, Evanthia and Suderland, Martin},
  journal={Discrete \& Computational Geometry},
  volume={72},
  number={3},
  pages={1304--1332},
  year={2024},
  publisher={Springer}
}

@misc{papadopoulou2026voronoidiagramlinesmathbbr3,
      title={The Voronoi Diagram of Four Lines in $\mathbb{R}^3$}, 
      author={Evanthia Papadopoulou and Zeyu Wang},
      year={2026},
      eprint={2603.19836},
      archivePrefix={arXiv},
      primaryClass={cs.CG},
      url={https://arxiv.org/abs/2603.19836}, 
}

@article{lee2022robust,
  title={Robust construction of Voronoi diagrams of spherical balls in three-dimensional space},
  author={Lee, Mokwon and Sugihara, Kokichi and Kim, Deok-Soo},
  journal={Computer-Aided Design},
  volume={152},
  pages={103374},
  year={2022},
  publisher={Elsevier}
}

@inproceedings{emiris2006predicates,
  title={The predicates for the Voronoi diagram of ellipses},
  author={Emiris, Ioannis Z and Tsigaridas, Elias P and Tzoumas, George M},
  booktitle={Proceedings of the twenty-second annual symposium on Computational geometry},
  pages={227--236},
  year={2006}
}

@article{cifuentes2022voronoi,
  title={Voronoi cells of varieties},
  author={Cifuentes, Diego and Ranestad, Kristian and Sturmfels, Bernd and Weinstein, Madeleine},
  journal={Journal of symbolic computation},
  volume={109},
  pages={351--366},
  year={2022},
  publisher={Elsevier}
}

\newpage
\appendix

\section{Computer Algebra and SMT Verification Scripts}\label{app:scripts}

To guarantee the exactness and reproducibility of our hybrid topological verification framework, we provide the complete scripts used for the algebraic computations. The scripts are written for Macaulay2 (for exact Gröbner basis and ideal saturation) and Z3 (for strict singly-exponential SMT witness sampling). 
Please note that while the radius is denoted as $R$ in the main text, it is represented by the variable \texttt{rad} in the provided computer algebra scripts due to variable naming constraints in the systems.

This appendix is organized according to the computational roles in the verification pipeline:
\ref{Mac:Affine} extracts affine candidate walls, \ref{Mac:Proj} checks projective behavior at infinity, \ref{Z3:SMT} samples rational witnesses in parameter space, and \ref{Mac:ZeroDim} performs instantiated exact smoothness checks.

\subsection{Extraction of the Affine Singularity Discriminant (Macaulay2)}\label{Mac:Affine}
This script extracts the exact bifurcation set $\Sigma_{\text{affine}}$ where the space curve loses its smoothness in the finite affine space.

\begin{lstlisting}[language=Ruby, caption=Macaulay2 script for Jacobian evaluation and discriminant extraction.]
-- 1. Define the ring containing parameters and spatial variables
-- c and s are cos and sin of the crossing angle between the two lines.
Rring = QQ[h, k, rad, c, s, x, y, z, u, v]

-- 2. Set distance functions and geometric constraints
-- Difference in distance to L1 and L2 (bisector of two lines)
F1 = y^2 - (x*s - y*c)^2 + 2*z*h - h^2

-- Difference in distance to L1 and circle C
F2 = -x^2 + 2*x*u + 2*y*v + 2*z*k - rad^2 - k^2

-- Condition for the projection point (u, v) on circle C
F3 = u^2 + v^2 - rad^2

-- Orthogonality condition for the shortest distance to the circle
F4 = y*u - x*v

-- Trigonometric identity constraint (crossing angle)
F5 = c^2 + s^2 - 1

-- 3. Generate the ideal of the curve
I = ideal(F1, F2, F3, F4, F5)

-- 4. Manually construct the 4x5 Jacobian matrix with respect to spatial variables (x,y,z,u,v)
-- Do not differentiate with respect to parameters (h,k,rad,c,s).
Jac = matrix {
  {diff(x,F1), diff(y,F1), diff(z,F1), diff(u,F1), diff(v,F1)},
  {diff(x,F2), diff(y,F2), diff(z,F2), diff(u,F2), diff(v,F2)},
  {diff(x,F3), diff(y,F3), diff(z,F3), diff(u,F3), diff(v,F3)},
  {diff(x,F4), diff(y,F4), diff(z,F4), diff(u,F4), diff(v,F4)}
}

-- 5. Singularity condition: Jacobian rank drops below 4 (all 4x4 minors are 0)
SingCond = minors(4, Jac)

-- 6. Construct the full singular ideal (exists on curve I and satisfies SingCond)
SingIdeal = I + SingCond

-- 7. Eliminate spatial variables (x,y,z,u,v) to derive relations only in parameters (h,k,rad,c,s)
DiscriminantIdeal = eliminate({x, y, z, u, v}, SingIdeal)

-- 8. Output final discriminant
gens DiscriminantIdeal
\end{lstlisting}

\subsection{Projective Closure and Saturation (Macaulay2)}\label{Mac:Proj}
This script illustrates the projective-closure computation and the identification of the point at infinity.

\begin{lstlisting}[language=Ruby, caption=Macaulay2 script for projective closure and computing the intersection at infinity.]
-- 1. Define the ring: Homogeneous variables X, Y, Z, W in projective space P^3
R = QQ[X, Y, Z, W]

-- 2. Set arbitrary rational parameters representing general position
h = 1; k = 2; rad = 3; 
c = 3/5; s = 4/5; -- cos(alpha), sin(alpha)

-- 3. Generators obtained by simple homogenization of affine polynomials (F1^h, F2^h)
F1h = Y^2 - (X*s - Y*c)^2 + 2*h*Z*W - h^2*W^2
F2h = (X^2 - 2*k*Z*W + (rad^2 + k^2)*W^2)^2 - 4*rad^2*W^2*(X^2 + Y^2)

-- 4. Generate the simple homogeneous ideal
I_homogenized = ideal(F1h, F2h)

-- 5. Saturation operation to obtain the true projective closure
-- Identifies all polynomials that belong to the ideal when multiplied by W^n, 
-- algebraically removing embedded components on the hyperplane at infinity (W=0).
I_proj = saturate(I_homogenized, ideal(W))

-- 6. Compute intersection ideal with the hyperplane at infinity (find true points at infinity)
I_infty = I_proj + ideal(W)

-- 7. Output properties of the solution set at infinity
-- A dimension of 1 means it is 0-dimensional (a finite set of points) in P^3.
dim I_infty
-- The degree should be 8, which is the total degree of the trisector.
degree I_infty

-- 8. Algebraic verification of coordinates at infinity
gens gb I_infty
\end{lstlisting}

\subsection{Witness Point Sampling (Z3 SMT Solver)}\label{Z3:SMT}
This script safely samples strict rational coordinates deeply embedded within each isolated topological cell, deliberately avoiding the singularity manifolds.

\begin{lstlisting}[language=Python, caption=Z3 script for rational witness extraction.]
from z3 import *

# 1. Declare Z3 Real variables
k, rad, t = Reals('k rad t')

# 2. Define the 6 topological connected components partitioning the parameter space
# Boundary conditions: k=0, k=1 (h=1), t=0, rad=0
regions = [
    {"name": "Component 1 (k > h, t > 0)", "conds": [k > 1, rad > 0, t > 0]},
    {"name": "Component 2 (k > h, t < 0)", "conds": [k > 1, rad > 0, t < 0]},
    {"name": "Component 3 (0 < k < h, t > 0)", "conds": [k > 0, k < 1, rad > 0, t > 0]},
    {"name": "Component 4 (0 < k < h, t < 0)", "conds": [k > 0, k < 1, rad > 0, t < 0]},
    {"name": "Component 5 (k < 0, t > 0)", "conds": [k < 0, rad > 0, t > 0]},
    {"name": "Component 6 (k < 0, t < 0)", "conds": [k < 0, rad > 0, t < 0]},
]

print("Starting Z3 Solver to extract witness points for each topological component...\n")

# 3. Run Solver for each region to extract a witness
for region in regions:
    s = Solver()
    
    # Add inequality constraints for the region
    for cond in region["conds"]:
        s.add(cond)
        
    # Check if a solution exists (sat)
    if s.check() == sat:
        m = s.model()
        
        # Extract the result as an exact fraction
        # (To prevent floating-point errors in Macaulay2 symbolic computations)
        k_val = m[k].as_fraction()
        rad_val = m[rad].as_fraction()
        t_val = m[t].as_fraction()
        
        print(f"--- {region['name']} ---")
        print(f"Found Witness: k = {k_val}, rad = {rad_val}, t = {t_val}")
        
        # Format that can be directly copied/pasted into a Macaulay2 script
        print("[Macaulay2 Input]")
        print(f"h = 1; k = {k_val}; rad = {rad_val}; t = {t_val};")
        print("-" * 50)
    else:
        print(f"--- {region['name']} ---")
        print("No witness found (Region is mathematically empty).")
        print("-" * 50)
\end{lstlisting}

\subsection{Exact Verification that the Affine Singular Locus is Empty (Macaulay2)}\label{Mac:ZeroDim}
In this script, we substitute the SMT-extracted witness point into the parametric system. While the geometric curve itself remains 1-dimensional in affine space, our objective is to prove smoothness. Therefore, we isolate the affine singular locus ($Sing_{affine}$) and demonstrate that it is shown to be empty, equivalently the ideal is the unit ideal, satisfying the Weak Nullstellensatz.

\begin{lstlisting}[language=Ruby, caption={Macaulay2 script verifying the absence of affine singularities. The script explicitly demonstrates that the singular ideal is trivial.}, label={lst:m2_singular}]
-- 1. Define the homogeneous coordinate ring over the rational field
R = QQ[X,Y,Z,W];

-- 2. Instantiate the SMT-extracted rational witness point
-- (Example from Component 3: k=1/4, R=1/8, t=1)
h = 1; k = 1/4; rad = 1/8; t = 1;
c = (1-t^2)/(1+t^2); s = (2*t)/(1+t^2);

-- 3. Define the instantiated homogeneous generators
F1h = Y^2 - (X*s - Y*c)^2 + 2*h*Z*W - h^2*W^2;
F2h = (X^2 - 2*k*Z*W + (rad^2 + k^2)*W^2)^2 - 4*rad^2*W^2*(X^2 + Y^2);

-- 4. Compute the projective closure of the space curve
I_proj = saturate(ideal(F1h, F2h), ideal(W));
-- Note: dim(I_proj) = 2, representing a 1D space curve in P^3.

-- 5. Construct the Projective Singular Locus
SingProj = I_proj + minors(2, jacobian(I_proj));

-- 6. Isolate the Affine Singular Locus by saturating out the hyperplane at infinity
SingAffine = saturate(SingProj, ideal(W));

-- 7. Weak Nullstellensatz Validation
-- If the affine singular locus is empty, the ideal must be trivial (ideal(1_R)).
SingAffine == ideal(1_R) 
\end{lstlisting}

\subsection{Extraction of the Infinity Discriminant (Macaulay2)}\label{Mac:InfinityDisc}
This script records the exact local computation at the unique point at infinity
$p_\infty=[0:0:1:0]$ in the normalized chart $Z=1$.
Its purpose is to extract the quadratic tangent-direction form
\[
E(X,Y)=X^2-2kQ(X,Y),
\qquad
Q(X,Y)=\frac12\bigl(s^2X^2-2sc\,XY-s^2Y^2\bigr),
\]
and to verify that its discriminant is
\[
\Delta_Q = 4ks^2(k-1).
\]
In particular, this script makes explicit the exact algebraic quantity governing the
candidate asymptotic degeneration locus at infinity.
Note that the radius parameter \texttt{rad} does not appear in the final discriminant,
since the tangent-cone computation depends only on the lowest-order local terms.

\begin{lstlisting}[language=Ruby, caption={Macaulay2 script for extracting the tangent-direction form at infinity and computing its discriminant.}, label={lst:m2_infinity_discriminant}]
-- 1. Work in the normalized parameter setting h = 1
-- Impose the trigonometric relation c^2 + s^2 = 1 at the ring level.
S = QQ[k, rad, c, s]/ideal(c^2 + s^2 - 1)

-- 2. Local affine chart around p_infinity = [0:0:1:0] with Z = 1
R = S[X, Y, W]

-- 3. Local form of the first homogenized bisector in the chart Z = 1
-- F1^h(X,Y,1,W) = Y^2 - (X*s - Y*c)^2 + 2W - W^2
F1loc = Y^2 - (X*s - Y*c)^2 + 2*W - W^2

-- 4. Extract the quadratic part of the local solution W = Q(X,Y) + higher order terms
-- Since W has order 2 locally, the term W^2 contributes only in order >= 4.
Q = ((X*s - Y*c)^2 - Y^2)/2

-- 5. Form the tangent-direction polynomial
-- E(X,Y) = X^2 - 2*k*Q(X,Y)
E = X^2 - 2*k*Q

-- 6. Display the quadratic form explicitly
E

-- 7. Read off the coefficients of
-- E = a*X^2 + b*X*Y + c2*Y^2
a  = 1 - k*s^2
b  = 2*k*s*c
c2 = k*s^2

-- 8. Compute and factor the discriminant
DeltaQ = b^2 - 4*a*c2
factor DeltaQ
\end{lstlisting}

\subsection{Generic three-line benchmark (Macaulay2)}

The following script records the generic three-line witness used in the benchmark section and verifies the emptiness of the affine singular locus.

\begin{lstlisting}[language=Ruby, caption={Macaulay2 script for the generic three-line benchmark.}, label={lst:m2_three_line_generic}]
-- 1. Define the homogeneous coordinate ring
R = QQ[X,Y,Z,W];

-- 2. Instantiate a generic rational witness point
a = 2; u = 1; v = 2; alpha = 1; beta = 1;

-- 3. Define the homogenized bisector equations
F12h = Z*W*(1+a^2) + a*X*Y;

F13h = ((u^2 + v^2)*X^2 + (alpha^2 + v^2)*Y^2 - (u^2 + alpha^2)*Z^2
     - 2*u*alpha*X*Y - 2*v*alpha*X*Z - 2*u*v*Y*Z
     + 2*beta*(u^2 + v^2)*X*W - 2*a*beta*(u^2 + v^2)*Y*W);

-- 4. Compute true projective closure
I_proj = saturate(ideal(F12h, F13h), ideal(W));

-- 5. Projective singular locus
SingProj = I_proj + minors(2, jacobian(I_proj));

-- 6. Affine singular locus
SingAffine = saturate(SingProj, ideal(W));

-- 7. Degree and emptiness check
degree I_proj
SingAffine == ideal(1_R)
\end{lstlisting}

\subsection{Degenerate cubic-plus-line three-line benchmark (Macaulay2)}

Finally, the following script verifies the classical degenerate case in which the three-line trisector splits into a nonsingular cubic and a line.

\begin{lstlisting}[language=Ruby, caption={Macaulay2 script for the degenerate cubic-plus-line three-line benchmark.}, label={lst:m2_three_line_degenerate}]
-- 1. Define the ring over the rational field
R = QQ[X,Y,Z,W]

-- 2. Exact degenerate witness
a = 2;

-- 3. Homogenized bisector equations
F12h = 2*X*Y + 5*Z*W;
F13h = (3*X^2 - 4*Y^2 + Z^2 - 4*X*Y + 4*X*Z + 20*X*W
     - 40*Y*W - 50*Z*W - 175*W^2);

-- 4. True projective closure
I_proj = saturate(ideal(F12h, F13h), ideal(W))

-- 5. Primary decomposition
comps = decompose I_proj
apply(comps, degree)

-- 6. Separate the cubic and the line
CubicIdeal = if degree comps_0 == 3 then comps_0 else comps_1
LineIdeal  = if degree comps_0 == 1 then comps_0 else comps_1

-- 7. Verify nonsingularity of the cubic
SingCubic = CubicIdeal + minors(2, jacobian(CubicIdeal))
dim saturate(SingCubic, ideal(W))

-- 8. Compute the intersection
IntersectionIdeal = CubicIdeal + LineIdeal
dim IntersectionIdeal

-- 9. Check intersection at infinity
dim(IntersectionIdeal + ideal(W))

-- 10. Dehomogenize and eliminate X, Y
Raffine = QQ[X,Y,Z]
AffineIntersection = substitute(IntersectionIdeal, {X=>X, Y=>Y, Z=>Z, W=>1})
ElimZ = eliminate({X, Y}, AffineIntersection)

-- 11. Output elimination polynomial
gens gb ElimZ
\end{lstlisting}

\end{document}